\pdfoutput=1
\documentclass[a4paper,USenglish,cleveref,autoref,thm-restate]{lipics-v2021}
\nolinenumbers
\hideLIPIcs

\usepackage{amsmath,amssymb}
\usepackage{xspace}
\usepackage{booktabs}
\usepackage{tikz}
\definecolor{wspdblue}{HTML}{1F77B4}
\definecolor{wspdorange}{HTML}{E66100}
\definecolor{wspdgreen}{HTML}{1A9850}

\newcommand{\eps}{\varepsilon}
\DeclareMathOperator{\diam}{diam}
\DeclareMathOperator{\dist}{dist}
\newcommand{\opt}{\mathrm{OPT}}
\newcommand{\lb}{\mathrm{LB}}
\newcommand{\R}{\mathbb{R}}
\newcommand{\W}{\mathcal{W}}
\newcommand{\ball}{\mathbb{B}}
\newcommand{\minwspd}{\textsc{minWSPD}\xspace}
\crefname{observation}{Observation}{Observations}
\Crefname{observation}{Observation}{Observations}

\title{Certified Small Well-Separated Pair Decompositions in Practice}
\titlerunning{Certified Small WSPDs in Practice}

\author{Guilherme D. da Fonseca}{Aix-Marseille Universit\'e and LIS, France}{}{https://orcid.org/0000-0002-9807-028X}{}
\authorrunning{G.\,D. da Fonseca}
\Copyright{Guilherme D. da Fonseca}


\supplementdetails{Software}{https://github.com/gfonsecabr/small-wspd-experiments}

\acknowledgements{Use of AI: This paper was prepared with substantial help from the AI assistant Claude (Anthropic). It took part in developing the algorithms and proofs, writing the code and running the experiments, checking the references, and writing the text. The author checked all results and is responsible for the content.}

\EventEditors{}
\EventNoEds{0}
\EventLongTitle{}
\EventShortTitle{}
\EventAcronym{}
\EventYear{}
\EventDate{}
\EventLocation{}
\EventLogo{}
\SeriesVolume{}
\ArticleNo{}

\begin{document}
\maketitle

\begin{abstract}
Given $n$ points and $\eps>0$, a well-separated pair decomposition (WSPD) is a collection of \emph{dumbbells} $\{X,Y\}$, where $X$ and $Y$ have diameter at most $\eps$ times their distance, such that every pair of points has one point in $X$ and the other in $Y$ for some dumbbell. In a \emph{partition} WSPD every pair is covered by exactly one dumbbell, while a \emph{cover} WSPD allows a pair to be covered by several. Callahan and Kosaraju give a partition WSPD with $O(n/\eps^d)$ dumbbells in $\R^d$, but many point sets have far smaller WSPDs of both kinds. For $d\ge2$, the known $O(1)$-approximations of the minimum size WSPD that preserve the separation $\eps$ have unspecified factors, above $10^7$ for $d=2$ by our estimates.

For cover WSPDs, we give an $O(n^2\log n)$-time \emph{sweep approximation algorithm} with a factor $40$ in the plane for $\eps \le 1/16$. It extends to $\R^d$ with a factor $2^{O(d)}$. The sweep is also a \emph{lower-bound tool}: its number of steps is a lower bound on the WSPD size. For partition WSPDs, we \emph{analyze quadtree and fair split tree constructions}: in the plane and for $\eps \le 1/5$, they give constant approximation factors, the smallest of which is $252$ for a randomly shifted and scaled quadtree.

In experiments with $\eps=1/20$, the number of sweep steps is at least $0.91\,\opt$ for $n=500$, where the optimum $\opt$ is known, and a local search raises it to a lower bound of $0.99\,\opt$.
A greedy variant of the sweep, which is still a certified $40$-approximation, is at most $1.28$ times the lower bound for $n=10000$.
Followed by a new local-search heuristic that shrinks the WSPD, it is within $0.02\%$ of the optimum on every instance with $n=500$, and at most $1.10$ times the lower bound for $n=10000$.
We then show how to convert the cover WSPDs into partition WSPDs.
For $n=10000$, the conversion gives on average $1.13$ times the lower bound, against $1.05$ for the cover WSPD itself and $1.55$ for the Callahan--Kosaraju partition WSPD.
\end{abstract}

\section{Introduction}
\label{sec:intro}

Let $P$ be a set of $n$ points in fixed dimension $d$, and let $\eps\in(0,1)$. We write $\|pq\|$ for the Euclidean distance between points $p$ and $q$, $\diam(X)$ for the diameter of a set $X\subseteq P$, and $\dist(X,Y)=\min_{x\in X,y\in Y}\|xy\|$. A \emph{dumbbell} $\{X,Y\}$ consists of two disjoint nonempty sets $X,Y\subseteq P$, its \emph{sides}, and it is \emph{$\eps$-separated} if $\max\{\diam(X),\diam(Y)\}\le \eps\cdot\dist(X,Y)$.

An \emph{$\eps$-WSPD} of $P$ is a collection $\W$ of $\eps$-separated dumbbells that together cover all $\binom n2$ \emph{pairs} of points, and $|\W|$ is its \emph{size}. A \emph{partition} WSPD covers every pair exactly once, and we use the term \emph{cover} WSPD to stress that pairs may be covered by multiple dumbbells.

WSPDs were introduced in~\cite{CK95} and are now a standard tool in computational geometry~\cite{NS07,Smi18}. The Callahan--Kosaraju (CK) construction~\cite{CK95} computes a partition $\eps$-WSPD of size $O(n/\eps^d)$ in $O(n\log n + n/\eps^d)$ time, with dumbbells given as pairs of tree nodes.
Several applications of WSPDs only use that every pair is covered at least once, so they work with cover WSPDs. Choosing one edge per dumbbell gives a $t$-spanner with $|\W|$ edges~\cite{CK93,NS07,Smi18}, and hence an approximate minimum spanning tree~\cite{CK93,Smi18}; such sparse edge sets also speed up solvers for geometric tour problems~\cite{FKKN24}. The approximate diameter~\cite{Smi18,Har11} is found by inspecting every dumbbell once. Other applications, such as the computation of $n$-body potential fields~\cite{CK95}, approximate distance selection~\cite{BS02,Smi18}, and the approximation of the sum of all pairwise distances~\cite{Har11}, need a partition WSPD.

The bound $O(n/\eps^d)$ is a worst-case bound, and many point sets have much smaller WSPDs. For example, the $n$ points $p_i=(i,(4/\eps)^i)$ have an $\eps$-WSPD with only $n-1$ dumbbells even in the partition version~\cite{BCHKLPRW26}. Since spanners and many other structures built from a WSPD inherit its size, a WSPD that is as small as possible is desirable.

The input of \minwspd consists of $P$ and $\eps$, and the goal is a cover $\eps$-WSPD of $P$ of minimum size, which we denote by $\opt_\eps$, or by $\opt$ if $\eps$ is clear from the context. A partition WSPD is in particular a cover WSPD. Therefore, every approximation factor with respect to $\opt_\eps$ also holds with respect to the partition optimum.

\subparagraph*{Previous work.}
\minwspd was introduced in~\cite{BCHKLPRW26}. It is NP-hard in the plane when $\eps$ may depend on $n$, and in general metric spaces no polynomial-time algorithm achieves a factor $n^c$, for some constant $c>0$, unless P${}={}$NP~\cite{BCHKLPRW26}. \Cref{tab:results} lists the known approximation algorithms. For $d\ge2$, the only one stated for the separation $\eps$ is based on the net-tree WSPD~\cite{HM06}, and its factor $2^{O(d)}$ is unspecified; our rough estimate from its proof is $10^{11}$ for $d=2$. Greedy set cover over all maximal $\eps$-separated dumbbells keeps the separation and gives an $O(\log n)$-approximation whenever their number is polynomial in $n$, which holds on the line but is not known for $d\ge2$~\cite{BCHKLPRW26}. The constant-factor algorithm of~\cite{HRR25} works in every metric space but relaxes the separation (the constants $7$ and $33$ follow from its proof, see also~\cite{BCHKLPRW26}). In $\R^d$, the relaxation can be undone: a grid cuts each side of a $c\eps$-separated dumbbell into at most $\lceil c\sqrt d\,\rceil^d$ pieces whose diameter is at most $1/c$ times the larger diameter of the two sides, and the pieces form $\eps$-separated dumbbells. For $d=2$ and $\eps<1/33$, the factor becomes $7\cdot47^4\approx3.4\cdot10^7$. For $d\ge2$, the question of a $(1+\delta)$-approximation, or a simple $O(1)$-approximation, was raised in~\cite{BCHKLPRW26}. A bicriteria approximation scheme, which computes a $(1+\delta)\eps$-WSPD of size at most $(1+\delta)\opt_\eps$ in $\R^d$ and in doubling metrics, is given in~\cite{Fon26bicriteria}, but its running time is doubly exponential in $d/\delta$.

\subparagraph*{Our results.}
\Cref{tab:results} places our algorithms among the known approximation results.
\begin{enumerate}
\item \textbf{Improved cover WSPD approximation factors} (\cref{thm:sweep}). For points in the plane and $\eps\le 1/16$, a sweep algorithm computes a $40$-approximation in $O(n^2\log n)$ time and $O(n^2)$ space. It extends to $\R^d$ for every fixed $d$ (\cref{thm:sweepd}) with the same time and space bounds, with a factor $2048$ for $d=3$ and $2^{O(d)}$ in general. It generalizes the $3$-approximation of~\cite{BCHKLPRW26} for $d=1$ and is a simple $O(1)$-approximation, as asked there.

\item \textbf{Improved partition WSPD approximation factors} (\cref{sec:partition}). For points in the plane and $\eps\le1/5$, the CK construction on a fair split tree~\cite{CK95} outputs at most $1024\,\opt$ dumbbells in $O(n\log n+\opt)$ time. The same size bound holds for the two variants of our experiments: \emph{CK box} accepts two tree nodes if their bounding boxes are $\eps$-separated, and \emph{CK exact}, or just \emph{CK}, if their point sets are. In the same time, the maximal separated pairs of same-level quadtree cells, with separation tested on the cells, form a partition $\eps$-WSPD of size at most $529\,\opt$, and of expected size at most $252\,\opt$ for a randomly shifted and scaled quadtree. The deterministic bounds extend to $\R^d$.

\item \textbf{Small cover WSPD in practice} (\cref{sec:experiments}). A greedy variant of the sweep keeps the approximation factor $40$ and runs in $O(n^3)$ time (\cref{thm:greedy}). The number of steps of the sweep is a lower bound on $\opt$. A local search on the anchors of the steps enlarges it to a lower bound $\lb$ (\cref{sec:lb}), which certifies the quality of any WSPD. For $\eps=1/20$ and $n=500$, where $\opt$ is known, the greedy sweep is on average $1.03\,\opt$, against $1.17\,\opt$ for the CK WSPD, and the lower bound is at least $0.99\,\opt$. For $n=10000$, where $\opt$ is out of reach, the greedy sweep and the CK WSPD are on average $1.22\,\lb$ and $1.55\,\lb$.
We present a local-search heuristic (\cref{sec:reshape}) that post-processes any given WSPD with a single move, \emph{reshape}, never increasing its size. For $n=500$, it brings the sweep to within $0.02\%$ of $\opt$, and it reaches $\opt$ on $13$ of the $22$ instances. For $n=10000$, the sweep followed by local search is at most $1.10\,\lb$ and on average $1.05\,\lb$.

\item \textbf{Small partition WSPD in practice} (\cref{sec:experiments}). Converting the cover WSPD of the sweep followed by local search into a partition WSPD gives on average $1.006\,\opt$ for $n=500$ and $1.13\,\lb$ for $n=10000$, against $1.17\,\opt$ and $1.55\,\lb$ for CK.
\end{enumerate}

\begin{table}[htb]
\centering\small\setlength{\tabcolsep}{4pt}
\caption{Approximation algorithms for \minwspd. \emph{Both}: outputs a partition WSPD. $^\ast$: expected.}
\label{tab:results}
\begin{tabular}{llllll}
\toprule
Setting & Problem & Size & Separation & Time & Reference \\
\midrule
$\R^1$ & cover & $(1+\delta)\,\opt_\eps$ & $\eps$ & $n^{O(1/\delta^2)}$ & \cite{BCHKLPRW26} \\
$\R^1$ & cover & $3\,\opt_\eps$ & $\eps$ & $O(n\log n)$ & \cite{BCHKLPRW26} \\
$\R^1$ & both & $10\,\opt_\eps$ & $\eps$ & polynomial & \cite{BCHKLPRW26} \\
$\R^d$ & both & $2^{O(d)}\opt_\eps$ & $\eps$ & $2^{O(d)}(n\log n+\opt)^{\ast}$ & \cite{BCHKLPRW26,HM06} \\
$\R^d$ & cover & $(1+\delta)\,\opt_\eps$ & $(1+\delta)\eps$ & $f(d,\delta)\,n^6$ & \cite{Fon26bicriteria} \\
Any metric & cover & $7\,\opt_{\eps}$ & $33\eps$ & polynomial & \cite{HRR25} \\
\addlinespace[2pt]
$\R^2$, $\eps\le1/16$ & cover & $40\,\opt_\eps$ & $\eps$ & $O(n^2\log n)$ & \cref{thm:sweep} \\
$\R^2$, $\eps\le1/5$ & both & $1024\,\opt_\eps$ & $\eps$ & $O(n\log n+\opt)$ & \cite{CK95}, \cref{sec:partition} \\
$\R^2$, $\eps\le1/5$ & both & $529\,\opt_\eps$ & $\eps$ & $O(n\log n+\opt)$ & \cref{sec:partition} \\
$\R^2$, $\eps\le1/5$ & both & $252\,\opt_\eps{}^{\ast}$ & $\eps$ & $O(n\log n+\opt)$ & \cref{sec:partition} \\
$\R^d$, $\eps\le1/16$ & cover & $2^{O(d)}\opt_\eps$ & $\eps$ & $O(n^2\log n)$ & \cref{thm:sweepd} \\
\bottomrule
\end{tabular}
\end{table}

\Cref{sec:sweep} presents the sweep and its analysis, \cref{sec:ls} a greedy variant of the sweep and the local search for smaller WSPDs and larger lower bounds, and \cref{sec:higher} the extension of the sweep to higher dimensions. \Cref{sec:experiments} reports the experiments, in the plane and, in \cref{sec:exp3d}, in $\R^3$. \Cref{sec:largeeps} extends the sweep to larger $\eps$. The partition WSPD constructions are analyzed in \cref{sec:partition}.

\section{Sweep Algorithm}
\label{sec:sweep}

This section presents the sweep $40$-approximation algorithm for points in the plane. Let $\prec$ denote the lexicographic order on $\R^2$ (by first coordinate, then by second coordinate), and write $p\preceq q$ if $p\prec q$ or $p=q$. We write $\ball(c,r)$ for the closed disk of radius $r$ centered at~$c$.

\subparagraph*{Pieces.}
For a disk $\ball$ of radius $\rho$, we use a fixed partition of $\ball$ into $8$ pieces: the concentric disk of radius $0.434\rho$, and the $7$ sectors of angle $2\pi/7$ of the annulus between radii $0.434\rho$ and $\rho$. For the half-disk $\ball^+$, which consists of the points of $\ball$ whose first coordinate is at least that of its center, we use $5$ pieces: the concentric half-disk of radius $0.4\rho$, and the $4$ sectors of angle $\pi/4$ of the half-annulus between radii $0.4\rho$ and $\rho$. Points on the boundary between two pieces are assigned to one of them arbitrarily, so the pieces partition the disk and the half-disk. The inner radii are chosen so that all pieces have roughly the same diameter.

\begin{lemma}\label{lem:pieces}
Every piece of the disk partition has diameter at most $0.868\rho$, and every piece of the half-disk partition has diameter at most $0.8\rho$.
\end{lemma}
\begin{proof}
Assume $\rho=1$. The diameter of a set is attained by two extreme points of its convex hull. For an annular sector with angle less than $\pi$, these are the two inner corners and the points of the outer arc. Two points of the outer arc are at distance at most $2\sin(\varphi/2)$, and the distance from an inner corner to a point of the outer arc increases with the angle between them. Hence an annular sector with inner radius $a<1$, outer radius $1$, and angle $\varphi\le\pi/2$ has diameter $\max\{2\sin(\varphi/2),\sqrt{1+a^2-2a\cos\varphi},1-a\}$.
For $\varphi=2\pi/7$ and $a=0.434$, these values are $0.8678$, $0.804$, and $0.566$. For $\varphi=\pi/4$ and $a=0.4$, they are $0.765$, $0.771$, and $0.6$. The central pieces have diameters $0.868$ and $0.8$.
\end{proof}

\begin{figure}[htb]
\centering
\begin{tikzpicture}[font=\small]
\fill[black!8] (0,-1.75) arc[start angle=-90, end angle=-270, radius=1.75] -- cycle;
\fill[wspdblue!25] (0.495,-0.495) arc[start angle=-45, end angle=0, radius=0.7] -- (1.75,0) arc[start angle=0, end angle=-45, radius=1.75] -- cycle;
\fill[wspdblue!25] (7.406,1.674) arc[start angle=141.429, end angle=192.857, radius=0.759] -- (6.294,0.811) arc[start angle=192.857, end angle=141.429, radius=1.75] -- cycle;
\draw[black!45] (0,0) circle[radius=1.75];
\draw[thick] (0,-1.75) arc[start angle=-90, end angle=90, radius=1.75] -- cycle;
\draw (0,-0.7) arc[start angle=-90, end angle=90, radius=0.7];
\draw (0.495,-0.495) -- (1.237,-1.237);
\draw (0.7,0) -- (1.75,0);
\draw (0.495,0.495) -- (1.237,1.237);
\draw[thick] (8,1.2) circle[radius=1.75];
\draw (8,1.2) circle[radius=0.759];
\draw (8,1.959) -- (8,2.95);
\draw (7.406,1.674) -- (6.632,2.291);
\draw (7.26,1.031) -- (6.294,0.811);
\draw (7.67,0.516) -- (7.241,-0.377);
\draw (8.33,0.516) -- (8.759,-0.377);
\draw (8.74,1.031) -- (9.706,0.811);
\draw (8.594,1.674) -- (9.368,2.291);
\draw[black!60] (0,0) -- (3.68,0.552) (4.32,0.648) -- (8,1.2);
\draw[black!60] (3.603,0.662) -- (3.757,0.442);
\draw[black!60] (4.243,0.758) -- (4.397,0.538);
\node[black!70, above=2pt] at (4,0.6) {$D$};
\fill[black] (0.441,-1.185) circle[radius=1.3pt];
\fill[black!35] (-0.435,-1.469) circle[radius=1.3pt];
\fill[black] (1.096,0.908) circle[radius=1.3pt];
\fill[black!35] (-0.51,1.13) circle[radius=1.3pt];
\fill[wspdblue] (1.515,-0.543) circle[radius=1.3pt];
\fill[black] (0.167,-0.403) circle[radius=1.3pt];
\fill[black] (0.955,1.33) circle[radius=1.3pt];
\fill[black!35] (-1.37,-0.532) circle[radius=1.3pt];
\fill[black] (0.465,-0.185) circle[radius=1.3pt];
\fill[wspdblue] (1.009,-0.262) circle[radius=1.3pt];
\fill[black] (0.163,-1.059) circle[radius=1.3pt];
\fill[black!35] (-0.808,0.98) circle[radius=1.3pt];
\fill[wspdblue] (0.848,-0.545) circle[radius=1.3pt];
\fill[black!35] (-0.188,1.317) circle[radius=1.3pt];
\fill[black!35] (-0.612,-0.678) circle[radius=1.3pt];
\fill[black] (0.367,1.032) circle[radius=1.3pt];
\fill[black] (0.512,0.712) circle[radius=1.3pt];
\fill[black] (0.136,0.565) circle[radius=1.3pt];
\fill[black] (1.5,0.482) circle[radius=1.3pt];
\fill[black] (8.815,-0.103) circle[radius=1.3pt];
\fill[wspdblue] (6.833,2.007) circle[radius=1.3pt];
\fill[black] (8.412,1.444) circle[radius=1.3pt];
\fill[black] (7.878,2.471) circle[radius=1.3pt];
\fill[wspdblue] (6.515,1.323) circle[radius=1.3pt];
\fill[wspdblue] (6.907,1.293) circle[radius=1.3pt];
\fill[black] (8.199,2.704) circle[radius=1.3pt];
\fill[black] (8.37,2.115) circle[radius=1.3pt];
\fill[black] (9.039,1.387) circle[radius=1.3pt];
\fill[black] (6.625,0.575) circle[radius=1.3pt];
\fill[black] (8.892,0.494) circle[radius=1.3pt];
\fill[black] (7.116,0.724) circle[radius=1.3pt];
\fill[black] (7.075,2.211) circle[radius=1.3pt];
\fill[black] (7.407,1.96) circle[radius=1.3pt];
\fill[black] (7.621,-0.073) circle[radius=1.3pt];
\fill[wspdblue] (6.546,1.679) circle[radius=1.3pt];
\fill[black] (8.015,0.004) circle[radius=1.3pt];
\fill (0,0) circle[radius=1.9pt] node[left=1pt] {$p$};
\fill (8,1.2) circle[radius=1.9pt] node[above right=0pt] {$q$};
\node[wspdblue!70!black] at (1.811,-0.75) {$A_i$};
\node[wspdblue!70!black] at (6.055,1.644) {$C_k$};
\node[black, align=center, font=\footnotesize] at (-1,0.12) {$z\prec p$};
\node[below] at (0,-1.8) {$\ball^+(p,\eps D)$};
\node[below] at (8,-0.6) {$\ball(q,\eps D)$};
\end{tikzpicture}
\caption{A step of the sweep with anchor $\{p,q\}$. The step outputs a dumbbell $\{A_i,C_k\}$, such as the one highlighted, for every two nonempty pieces.}
\label{fig:step}
\end{figure}
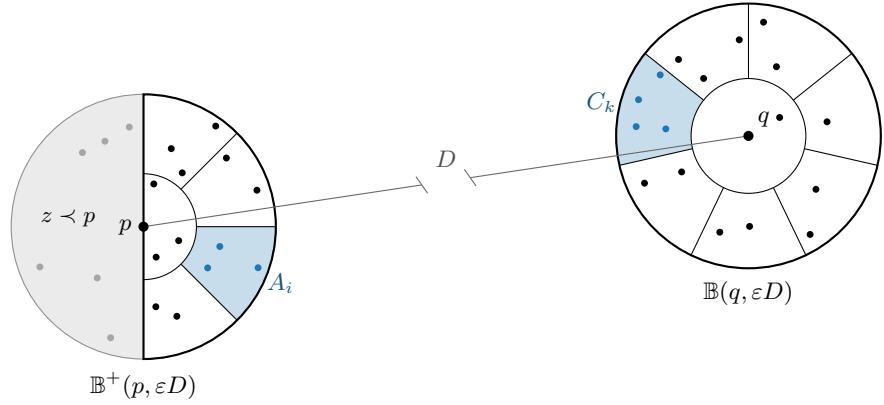

\subparagraph*{Algorithm.}
Initially, no pair is covered. The algorithm processes the points $p\in P$ in the order $\prec$, and for each such $p$ the points $q\succ p$ in arbitrary order. If the pair $\{p,q\}$ is not yet covered, the algorithm performs a \emph{step} with \emph{anchor} $\{p,q\}$, as follows. Let $D=\|pq\|$. Let $A_1,\dots,A_5$ be the intersections of $\{z\in P: z\succeq p\}$ with the pieces of the half-disk $\ball^+(p,\eps D)$, and let $C_1,\dots,C_8$ be the intersections of $P$ with the pieces of the disk $\ball(q,\eps D)$. The step outputs every dumbbell $\{A_i,C_k\}$ whose sides are both nonempty (\cref{fig:step}), at most $5\cdot8=40$ dumbbells. We call these dumbbells the \emph{fixed collection} of the step.

\begin{observation}[Locality]\label{obs:local}
If an $\eps$-separated dumbbell $\{X,Y\}$ covers $\{x,y\}$ and $\ell=\dist(X,Y)$, then $\ell\le\|xy\|\le(1+2\eps)\ell$, and each side lies in the disk of radius $\eps\ell$ centered at any of its points.
\end{observation}

\begin{theorem}\label{thm:sweep}
For $P\subset\R^2$ and $\eps\le1/16$, the sweep performs at most $\opt_\eps$ steps and outputs an $\eps$-WSPD of size at most $40\,\opt_\eps$ in $O(n^2\log n)$ time and $O(n^2)$ space.
\end{theorem}
\begin{proof}
\emph{Validity.} The disks $\ball(p,\eps D)$ and $\ball(q,\eps D)$ are disjoint since $\eps<1/2$. By \cref{lem:pieces}, $\diam(A_i),\diam(C_k)\le0.868\,\eps D$, while $\dist(A_i,C_k)\ge(1-2\eps)D$. Since $0.868\le 1-2\eps$ for $\eps\le1/16$, every output dumbbell is $\eps$-separated. Every pair $\{p,q\}$ with $p\prec q$ is covered once $p$ is processed, since the step with anchor $\{p,q\}$ covers it ($p$ lies in some $A_i$, $q$ in some $C_k$).

\emph{Invariant.} By the validity argument and the processing order, when the step with anchor $\{p,q\}$ starts, every pair $\{x,y\}$ with $\min_\prec\{x,y\}\prec p$ is already covered.

\emph{Charging.} Fix a minimum $\eps$-WSPD $\W^*$. Charge the step with anchor $\{p,q\}$ to some $Z=\{X,Y\}\in\W^*$ covering $\{p,q\}$, with $p\in X$ and $q\in Y$. Let $\ell=\dist(X,Y)\le D$. By \cref{obs:local}, $X\subseteq \ball(p,\eps\ell)\subseteq \ball(p,\eps D)$ and $Y\subseteq \ball(q,\eps D)$.

Let $\{x,y\}$, with $x\in X$ and $y\in Y$, be a pair that is not covered before the step. By the invariant, $x\succeq p$, so $x\in \ball^+(p,\eps D)$ lies in some $A_i$, and $y$ lies in some $C_k$. Hence, after the step, every pair covered by $Z$ is covered. In particular, no later anchor is covered by $Z$, and $Z$ is charged only once. Hence distinct steps are charged to distinct dumbbells of $\W^*$, and there are at most $|\W^*|$ steps. Since each step outputs at most $5\cdot8=40$ dumbbells, the output has size at most $40\,|\W^*|$. The running time is analyzed below.
\end{proof}

The condition $\eps\le1/16$ keeps constants simple. With inner radius $\sin(\pi/7)\rho$ instead of $0.434\rho$, every disk piece has diameter at most $2\sin(\pi/7)\rho$, and the factor $40$ only needs $\eps\le1/2-\sin(\pi/7)\approx1/15.12$. For larger $\eps$, more pieces give larger factors (\cref{sec:largeeps}).

Two pairs are \emph{compatible} if a single $\eps$-separated dumbbell covers both of them.

\begin{observation}[Heredity]\label{obs:hered}
If $\{X,Y\}$ is $\eps$-separated, then so is $\{X',Y'\}$ for all nonempty $X'\subseteq X$ and $Y'\subseteq Y$. Hence $\{a,a'\}$ and $\{b,b'\}$ are compatible if and only if $\{\{a,b\},\{a',b'\}\}$ or $\{\{a,b'\},\{a',b\}\}$ has disjoint sides and is $\eps$-separated.
\end{observation}

For the step with anchor $\{p,q\}$, let $R_{pq}$ be the set of pairs $\{x,y\}$ that are uncovered when the step starts, are compatible with $\{p,q\}$, and satisfy $x,y\succeq p$.

\begin{observation}\label{obs:anchors}
If every step of the sweep, with anchor $\{p,q\}$, covers $R_{pq}$, then the anchors are pairwise incompatible, so the number of steps is at most $\opt_\eps$.
\end{observation}
\begin{proof}
Let $\{p',q'\}$ be the anchor of a step performed after the step with anchor $\{p,q\}$. Both $p'$ and $q'$ are $\succeq p$, and $\{p',q'\}$ was uncovered when the step with anchor $\{p,q\}$ was performed. If the two anchors were compatible, that step would have covered $\{p',q'\}$, a contradiction. Since no $\eps$-separated dumbbell covers two anchors, every $\eps$-WSPD has at least as many dumbbells as there are steps.
\end{proof}

The fixed collection covers $R_{pq}$ for every $\eps<1/2$ (the condition $\eps\le1/16$ is needed only for its dumbbells to be $\eps$-separated). Indeed, the charging argument uses only that $Z$ is an $\eps$-separated dumbbell that covers $\{p,q\}$, and if $\{x,y\}$ is compatible with $\{p,q\}$, then, by \cref{obs:hered}, $\{\{p,x\},\{q,y\}\}$ or $\{\{p,y\},\{q,x\}\}$ is such a dumbbell.

\begin{lemma}\label{lem:multiplicity}
Let $\eps\le 1/16$. Every pair is covered by the output dumbbells of at most $360$ steps of the sweep, and by at most one output dumbbell of each step.
\end{lemma}
\begin{proof}
\emph{At most one dumbbell per step.} Let $\{x,y\}$ be covered in the step with anchor $\{p,q\}$, and let $D=\|pq\|$. The disks $\ball(p,\eps D)$ and $\ball(q,\eps D)$ are disjoint, so one endpoint, say $x$, lies in $\ball(p,\eps D)$ and the other in $\ball(q,\eps D)$. The pieces partition each disk, so exactly one output dumbbell of the step covers $\{x,y\}$.

\emph{At most $360$ steps.} Let $L=\|xy\|$. We split the steps whose output covers $\{x,y\}$ into two classes, according to whether $x$ lies in the disk around the first or the second point of the anchor, and we bound the size of each class by $180$. Consider the class of steps with anchors $\{p_j,q_j\}$ such that $x\in \ball(p_j,\eps D_j)$ and $y\in \ball(q_j,\eps D_j)$, where $D_j=\|p_jq_j\|$. By the triangle inequality, $(1-2\eps)D_j\le L\le(1+2\eps)D_j$. Therefore, all $p_j$ lie in $\ball(x,r)$ and all $q_j$ lie in $\ball(y,r)$, where $r=\eps L/(1-2\eps)$. Since $2r<L$ for $\eps<1/4$, these two disks are disjoint, so the sets $\{p_j,p_k\}$ and $\{q_j,q_k\}$ below are disjoint. For $j\ne k$, the anchors are incompatible by \cref{obs:anchors}, so the dumbbell $\{\{p_j,p_k\},\{q_j,q_k\}\}$ is not $\eps$-separated. Since its distance is at least $L-2r$,
\[
  \max\{\|p_jp_k\|,\|q_jq_k\|\}>s,\qquad\text{where } s=\eps(L-2r)=r\,(1-4\eps)\ge 3r/4 .
\]
For $j\ne k$, if the products $\ball(p_j,s/2)\times\ball(q_j,s/2)$ and $\ball(p_k,s/2)\times\ball(q_k,s/2)$ of disks shared a point $(u,v)$, then $\|p_jp_k\|\le\|p_ju\|+\|up_k\|\le s$ and, similarly, $\|q_jq_k\|\le s$. Hence these products are pairwise disjoint. They are contained in $\ball(x,r+s/2)\times\ball(y,r+s/2)$. The measure of a product of two disks is the product of their areas, so comparing areas, the class has at most $\bigl((r+s/2)/(s/2)\bigr)^4=(1+2r/s)^4\le(11/3)^4<181$ steps.
\end{proof}

\begin{proof}[Proof of the running time in \cref{thm:sweep}]
We store the covered pairs in an $n\times n$ table. For every point $c\in P$, we also store the list of the points of $P$ sorted by distance from~$c$. These lists take $O(n^2\log n)$ time to build and $O(n^2)$ space. The algorithm inspects each pair once, when its $\prec$-smaller endpoint is processed, and checks the table in $O(1)$ time. This takes $O(n^2)$ time in total. When the algorithm starts processing $p$, it extracts from the list of $p$ the points $z\succeq p$, in $O(n)$ time.
In a step with anchor $\{p,q\}$, the points of $\ball(q,\eps D)$ form a prefix of the list of $q$, and the points $z\succeq p$ of $\ball(p,\eps D)$ form a prefix of this extracted list. The step scans these two prefixes and assigns each point to its piece in $O(1)$ time (in the real RAM model, since the sector boundaries have irrational slopes), so it takes $O(1+m)$ time, where $m$ is the number of scanned points. Every scanned point lies in an output dumbbell of the step: a point of $A_i$ together with the piece $C_k$ that contains $q$, and a point of $C_k$ together with the piece $A_i$ that contains $p$. Hence $m\le\sum(|A|+|C|)$, where the sum is over the output dumbbells $\{A,C\}$ of the step. Marking the pairs covered by $\{A,C\}$ takes $O(|A|\cdot|C|)$ time, and $|A|+|C|\le|A|\cdot|C|+1$. Since a step outputs at most $40$ dumbbells, all steps together take time linear in the number of steps plus the total number of pair markings. There are at most $\binom n2$ steps, and by \cref{lem:multiplicity} every pair is marked at most $360$ times, so all steps take $O(n^2)$ time. Sorting dominates the time, and the table and the lists the space.
\end{proof}

\subparagraph*{Dependence on $\eps$.} The bound of \cref{thm:sweep} can also be stated in terms of $\eps$. By \cref{obs:anchors}, the number of steps is at most $\opt_\eps=O(n/\eps^2)$ by the CK construction~\cite{CK95}. Hence, even if every step scans all of $P$ without sorting, the sweep takes $O(n^2/\eps^2)$ time. Together with \cref{thm:sweep}, this gives $O(n^2\min\{\log n,1/\eps^2\})$ time.

\section{Practical Algorithms}
\label{sec:ls}

We present a greedy variant of the sweep (\cref{sec:greedy}), and we use local search in two ways: to reduce the size of a given WSPD (\cref{sec:reshape}), and to enlarge the lower bound given by the steps of the sweep (\cref{sec:lb}). They give large practical improvements, but do not improve the approximation bounds.

\subsection{Greedy Sweep}
\label{sec:greedy}

The proofs of validity and of the approximation factor in \cref{thm:sweep} use only one property of a step with anchor $\{p,q\}$: it covers $R_{pq}$. This set contains $\{p,q\}$ itself and, by the invariant, every uncovered pair covered by an optimal dumbbell that covers $\{p,q\}$. Therefore, these two proofs hold verbatim if each step outputs \emph{any} collection of $\eps$-separated dumbbells covering $R_{pq}$ whose size is at most that of the fixed collection. The running-time analysis does not carry over, and we analyze the greedy variant separately below.

The sweep of \cref{thm:sweep}, which we call the \emph{fixed sweep}, is of theoretical interest only: every step outputs a dumbbell for every two nonempty pieces, even if a single dumbbell covers $R_{pq}$. The \emph{greedy sweep} keeps the proof but not the collection. It covers $R_{pq}$ with dumbbells grown from pairs of $R_{pq}$, and falls back to the fixed collection only when that collection is smaller. To \emph{grow} a dumbbell from a \emph{seed} pair $\{a,b\}$, we start from $\{\{a\},\{b\}\}$ and repeatedly add a point that \emph{fits}, that is, a point whose distance to every point of one side is at most $\eps\|ab\|$ and whose distance to every point of the other side is at least $\|ab\|$; the point joins the first side. The dumbbell remains $\eps$-separated, with $\{a,b\}$ as a closest cross pair. Among the points that fit, we choose one that maximizes the number of newly covered pairs of $R_{pq}$, breaking ties by the number of newly covered pairs overall, and we stop when no point fits. If $|R_{pq}|\le 16$, the step grows one dumbbell from each pair of $R_{pq}$ and covers $R_{pq}$ by greedy set cover over these candidates. Otherwise, it grows a dumbbell from $\{p,q\}$, and then repeatedly grows a dumbbell from the closest pair of $R_{pq}$ that is not yet covered, until $R_{pq}$ is covered. The second rule avoids growing $|R_{pq}|$ dumbbells when $R_{pq}$ is large.

For $\eps\le1/16$, the number of dumbbells grown in a step is capped at $40$. If $R_{pq}$ is not covered by then, the step outputs the fixed collection. The cap does not change the output: a greedy cover with more than $40$ dumbbells is larger than the fixed collection, which has at most $40$ dumbbells, so the step would output the fixed collection anyway. Finally, the greedy sweep deletes, one at a time, the dumbbells whose pairs are all covered by other dumbbells.

\begin{theorem}\label{thm:greedy}
For $P\subset\R^2$ and $\eps\le1/16$, the greedy sweep performs at most $\opt_\eps$ steps and outputs an $\eps$-WSPD of size at most $40\,\opt_\eps$ in $O(n^3)$ time and $O(n^2)$ space.
\end{theorem}
\begin{proof}
Grown dumbbells are $\eps$-separated by construction, and every step covers $R_{pq}$ with at most $40$ dumbbells, so the size bounds follow as explained above. The running time is proved below.
\end{proof}

We prove the time and space bounds of \cref{thm:greedy}. The proof holds in every dimension. Let $\lambda=17/8$. For a step with anchor $\{p,q\}$ and $D=\|pq\|$, let $B_p=P\cap\ball(p,\lambda\eps D)$ and $B_q=P\cap\ball(q,\lambda\eps D)$. All the work of the step takes place in these two sets. Indeed, every seed $\{a,b\}$ is $\{p,q\}$ or a pair of $R_{pq}$, so some $\eps$-separated dumbbell covers both $\{a,b\}$ and $\{p,q\}$, and its distance is at most $D$. By \cref{obs:local}, we may name the seed so that $a\in\ball(p,\eps D)$, $b\in\ball(q,\eps D)$, and $\|ab\|\le(1+2\eps)D$. A point that fits is within distance $\eps\|ab\|$ of $a$ or $b$, hence within distance $\eps(2+2\eps)D\le\lambda\eps D$ of $p$ or~$q$. Therefore, every pair of $R_{pq}$ and every dumbbell that the step grows or outputs has one side in $B_p$ and the other in $B_q$. The next lemma extends the second part of \cref{lem:multiplicity} to these larger disks.

\begin{lemma}\label{lem:greedymult}
Let $P\subset\R^d$, $\eps\le1/16$, and $\lambda=17/8$. Every pair $\{x,y\}$ satisfies $x\in\ball(p,\lambda\eps D)$ and $y\in\ball(q,\lambda\eps D)$, where $D=\|pq\|$, for the anchors $\{p,q\}$ of at most $2\cdot11^{2d}$ steps of the greedy sweep, in one of the two orders of $p$ and $q$.
\end{lemma}
\begin{proof}
Let $L=\|xy\|$. Consider the steps with anchors $\{p_j,q_j\}$ such that $x\in\ball(p_j,\lambda\eps D_j)$ and $y\in\ball(q_j,\lambda\eps D_j)$, where $D_j=\|p_jq_j\|$; the steps with the roles of $x$ and $y$ exchanged are counted in the same way. By the triangle inequality, $L\ge(1-2\lambda\eps)D_j$, so $p_j\in\ball(x,r)$ and $q_j\in\ball(y,r)$, where $r=\lambda\eps L/(1-2\lambda\eps)$. Since $\lambda\eps\le17/128<1/4$, we have $2r<L$, so these two balls are disjoint. The greedy sweep covers $R_{pq}$ in every step, so by \cref{obs:anchors}, the dumbbell $\{\{p_j,p_k\},\{q_j,q_k\}\}$ is not $\eps$-separated for $j\ne k$. Its distance is at least $L-2r$, so
\[
  \max\{\|p_jp_k\|,\|q_jq_k\|\}>s,\qquad\text{where } s=\eps(L-2r)=\frac{r\,(1-4\lambda\eps)}{\lambda}\ge\frac{15}{68}\,r .
\]
Hence the products $\ball(p_j,s/2)\times\ball(q_j,s/2)\subset\R^{2d}$ are pairwise disjoint. They are contained in $\ball(x,r+s/2)\times\ball(y,r+s/2)$. Comparing volumes, there are at most $(1+2r/s)^{2d}\le(151/15)^{2d}<11^{2d}$ such steps.
\end{proof}

\begin{proof}[Proof of the running time in \cref{thm:greedy}]
We store the covered pairs and the sorted lists as in the proof of \cref{thm:sweep}, in $O(n^2\log n)$ time, and inspect every pair once. Consider a step, and let $N=|B_p|+|B_q|$ and $M=|B_p|\cdot|B_q|$. Since $\lambda\eps<1/2$, the sets $B_p$ and $B_q$ are disjoint prefixes of the lists of $p$ and $q$. Enumerating $R_{pq}$ takes $O(M)$ time, and we mark its pairs in the table, so that membership in $R_{pq}$ is tested in $O(1)$ time. To grow a dumbbell, we maintain for every point of $B_p\cup B_q$ whether it fits and how many pairs it would newly cover. When a point joins a side, these values change only by its own contribution, so they are updated in $O(N)$ time, and a dumbbell is grown in $O(N^2)$ time. Its seed is found by scanning $R_{pq}$ in $O(M)$ time. The greedy set cover for $|R_{pq}|\le16$ takes $O(N)$ time, and marking each of the at most $40$ output dumbbells takes $O(M)$ time. Since $p\in B_p$ and $q\in B_q$, we have $N\le 2M$, and $N\le n$, so the step takes $O(N^2+M)=O(nM)$ time. By \cref{lem:greedymult}, every pair contributes to $M$ in $O(1)$ steps, so the sum of $M$ over all steps is $O(n^2)$. Hence the greedy sweep runs in $O(n^3)$ time, and the table and the lists take $O(n^2)$ space.
\end{proof}

\subsection{Local Search for Smaller WSPDs}
\label{sec:reshape}

The sweep of \cref{sec:sweep} never revisits a dumbbell once it is output. We now describe a local-search heuristic that tries to reduce the size of any given $\eps$-WSPD, for instance the CK WSPD or the output of the sweep. It uses a single move, \emph{reshape}, and never increases the size. It uses no solver: unlike exact and heuristic set-cover solvers, it never builds the set-cover instance of \minwspd, whose size limits these solvers in practice (\cref{sec:experiments}).

\subparagraph*{Reshape.}
A pair is \emph{private} to a dumbbell of $\W$ if that dumbbell is the only one covering it, and a dumbbell without private pairs is \emph{redundant} and can be deleted.
A point $p$ can \emph{join} the side $X$ of an $\eps$-separated dumbbell $\{X,Y\}$ if $p\notin X\cup Y$ and $\{X\cup\{p\},Y\}$ is still $\eps$-separated. Unlike the greedy sweep of \cref{sec:greedy}, which keeps the closest cross pair fixed, we allow $\dist(X\cup\{p\},Y)$ to be smaller than $\dist(X,Y)$. To \emph{grow} a dumbbell, we repeatedly add a point that can join one of its sides, choosing one that newly covers the most pairs private to other dumbbells (ties are broken at random), until no point can join.

The move (\cref{fig:reshape}) is applied to a \emph{target} dumbbell $T\in\W$. If $T=\{X,Y\}$ is redundant, delete it. Otherwise, replace $T$ by the \emph{subdumbbell spanned by its private pairs}, the dumbbell $\{X',Y'\}$ where $X'\subseteq X$ and $Y'\subseteq Y$ consist of the endpoints of these pairs, which is $\eps$-separated by \cref{obs:hered}. Grow it, and then delete, one at a time, the dumbbells that the new dumbbell makes redundant. Shrinking $T$ first releases the points that lie in no private pair of $T$, and with them most pairs that other dumbbells also cover, so the new dumbbell can grow in a different direction than the old one. Its growth rule favors pairs that are private to other dumbbells, so repeated reshapes tend to make dumbbells redundant.

\begin{figure}[htb]
\centering
\input{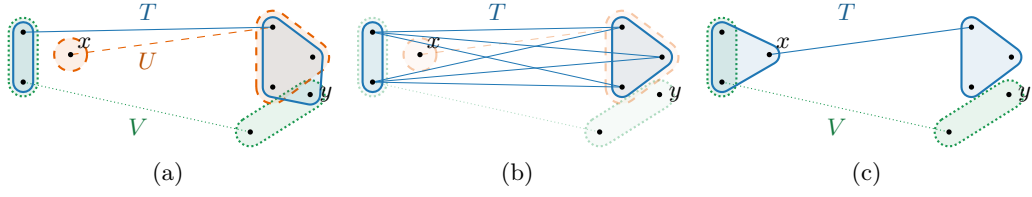}
\caption{A reshape of $T$ for $\eps=1/3$, showing three dumbbells, each as its two sides and a closest cross pair. (a)~With $y$ in $T$, $x$ cannot join $T$: the right side would be too wide for the smaller distance. (b)~The pairs private to $T$ (thin lines) span the subdumbbell, and $y$ is released. (c)~Free growth adds $x$, which covers the pairs private to $U$, so $U$ is deleted.}
\label{fig:reshape}
\end{figure}

The targets are taken in passes over $\W$, each in a new random order, until the time budget ends. Deleted dumbbells are skipped, and dumbbells created during a pass wait for the next one. At the end, the remaining redundant dumbbells are deleted.

\subparagraph*{Large neighborhood search.}
For comparison, we also use an integer-programming move. The \emph{large neighborhood search (LNS)} move removes a target $T$ and its $k$ nearest dumbbells, which leaves a set $U$ of uncovered pairs. It re-covers $U$ by a minimum number of dumbbells, chosen among the removed ones and, if $|U|\le u$, maximal $\eps$-separated dumbbells whose closest cross pair is in $U$, with CPLEX. The change is accepted only if it uses fewer dumbbells than it removes. LNS takes its targets in the same passes as reshape, and its neighborhoods and parameters are given in \cref{sec:setup}.

\begin{observation}\label{obs:ls}
Reshape and LNS keep $\W$ an $\eps$-WSPD and do not increase $|\W|$.
\end{observation}
\begin{proof}
In a reshape, the new dumbbell contains the subdumbbell spanned by the private pairs of $T$, so it covers every pair that only $T$ covered. It is $\eps$-separated by \cref{obs:hered} and because every added point keeps it $\eps$-separated. Only redundant dumbbells are deleted, and redundancy is checked again before each deletion, so every pair stays covered. A reshape replaces $T$ by one dumbbell and otherwise only deletes dumbbells, so $|\W|$ does not increase. LNS adds $\eps$-separated dumbbells that cover $U$, and only if the size decreases.
\end{proof}

\subparagraph*{Implementation.}
A reshape must decide which pairs are private, which dumbbells become redundant, and which points can join a side, each in constant time per pair or point. A table stores, for every pair, the number of dumbbells of $\W$ that cover it. Inserting or deleting $\{X,Y\}$ updates it in $O(|X|\cdot|Y|)$ time. A pair is private if its count is~$1$, so finding private pairs, testing redundancy, and computing the gain of a point take $O(1)$ time per pair.
Shrinking $T$ and deleting dumbbells only decrease counts, so a dumbbell becomes redundant only if the new $T$ covers one of its private pairs. To find it, we need the \emph{owner} of such a pair, the dumbbell to which it was private. Lists of the dumbbells containing each point are expensive: at $n=10000$, a point lies in hundreds to thousands of dumbbells. Instead, a second table stores, for every pair, the bitwise exclusive or of the indices of the dumbbells that cover it, updated in the same loop as the counts. Since $k\oplus k=0$, deleting a dumbbell cancels its insertion, so for a pair with count~$1$ the entry is the index of its owner. Hence the dumbbells that become redundant are found while $T$ is inserted, at no extra asymptotic cost. In a preliminary test with $n=10000$, this made reshape more than an order of magnitude faster than point-to-dumbbell lists.
A point $p$ that can join $X$ satisfies $\|px\|\le\eps\dist(X,Y)$ for every $x\in X$, so the candidates are the points within distance $\eps\dist(X,Y)$ of the initial $X$ or $Y$, found by range queries in a kd-tree. For $\eps<1/2$, each can join only the side it is close to. Diameters only increase and the distance between the sides only decreases, so a candidate that cannot join is discarded for good. For each candidate, we maintain its largest distance to its side, its smallest distance to the other side, and its gain, so that after each addition, updating them and testing separation exactly take constant time per candidate.
The points are relabeled in the order of a Hilbert curve, and the tables are packed upper triangles of $\binom n2$ entries of $4$ bytes, stored by rows, so the entries of the pairs of a dumbbell are close in memory.

\subsection{Local Search for Larger Lower Bounds}
\label{sec:lb}

The \emph{compatibility graph} $G_\eps$ has the pairs as vertices and an edge between any two compatible pairs. By the proof of \cref{obs:anchors}, every independent set $I$ of $G_\eps$ gives the lower bound $|I|\le\opt_\eps$. The anchors of the greedy sweep form such a set, and we call their number $\lb_0$. We enlarge this set by a local search for large independent sets, starting from the anchors. For the current set $I$, a pair is \emph{free} if it has no neighbor in $I$, and \emph{tight} for $e\in I$ if $e$ is its only neighbor in $I$. These terms and the $1$-for-$2$ swaps below follow the local search of Andrade, Resende, and Werneck~\cite{ARW12} for maximum independent sets. The search has three stages.
\begin{enumerate}
\item Add free pairs to $I$ until none is left.
\item Repeat passes over $I$. For each $e\in I$, if two pairs $f$ and $g$ that are tight for $e$ are incompatible with each other, replace $e$ by $f$ and $g$ (a \emph{$1$-for-$2$ swap}), and add the pairs that have become free. Stop after a pass without swaps, and let $\lb_1=|I|$.
\item Make $30$ more passes in which, in addition, a pair $e$ without a swap is replaced, with probability $1/2$, by a random pair that is tight for $e$, followed again by the free pairs. These \emph{plateau moves}~\cite{PH06,GLP08} do not change $|I|$, but they change which pairs are tight, so later passes find new swaps. The lower bound $\lb$ is $|I|$ at the end.
\end{enumerate}
A pair joins $I$ only if it is free or tight for the leaving pair, so $I$ stays independent and $\lb_0\le\lb_1\le\lb\le\opt_\eps$.

\subparagraph*{Implementation.}
The search keeps a table with, for every pair $w$, the number $t(w)$ of pairs of $I$ compatible with $w$. Hence $w$ is free if $t(w)=0$, and tight for $e$ if $t(w)=1$ and $w$ is compatible with $e$. If $\{x',y'\}$ is compatible with $e=\{x,y\}$, then, by \cref{obs:hered,obs:local} and after possibly exchanging $x'$ and $y'$, we have $x'\in\ball(x,\eps\|xy\|)$ and $y'\in\ball(y,\eps\|xy\|)$. As in the proof of \cref{thm:sweep}, these points are prefixes of the lists of the points sorted by distance from $x$ and from $y$, which gives the enumeration, and the update of the table when $e$ joins or leaves $I$, in $O(M_e)$ time, where $M_e$ is the product of the numbers of points in the two disks. In a pass, we skip $e$ if the dumbbell formed by the first points and the second points of the tight pairs of $e$, oriented as above, is $\eps$-separated: then all these pairs are pairwise compatible by \cref{obs:hered}. Otherwise, we test pairs of tight pairs for compatibility, at most $2\cdot10^5$ of them per $e$ and pass. Because of this limit, $\lb_1$ need not be a local optimum for $1$-for-$2$ swaps.

\begin{proposition}\label{prop:lbtime}
Let $\eps\le1/16$. The search takes $O(n^2\log n+kn^3)$ time and $O(n^2)$ space, where $k$ is the number of passes, and $k\le\lb_1-\lb_0+31$.
\end{proposition}
\begin{proof}
Sorting takes $O(n^2\log n)$ time, and the lists and the table take $O(n^2)$ space. The first stage scans all pairs once, and it adds pairs to $I$ without removing any, so its updates take $O(\sum_{e\in I}M_e)$ time for the final set $I$. The second part of the proof of \cref{lem:multiplicity} uses only that the anchors are pairwise incompatible. Applied to $I$, it shows that every pair lies in the product of the two disks of at most $360$ pairs of $I$, so $\sum_{e\in I}M_e=O(n^2)$ for every pairwise incompatible set $I$. In a pass, the pairs processed were in $I$ when the pass started, and the pairs added during the pass stay in $I$ until it ends, so enumerations and updates take $O(n^2)$ time in all. The test that skips $e$ takes time quadratic in the number $N_e$ of points in the two disks, and $N_e^2\le2nM_e$ as in the proof of \cref{thm:greedy}, which gives $O(n^3)$ per pass. The tests of pairs of tight pairs take $O(1)$ time per pair of $I$. Every pass of the second stage except the last increases $|I|$, so there are at most $\lb_1-\lb_0+31$ passes.
\end{proof}

In the experiments, stage~2 makes at most $13$ passes, and the time grows like $n^2$ (\cref{fig:lb}).

\section{Sweep in Higher Dimensions}
\label{sec:higher}

In this section, $P\subset\R^d$ for a fixed $d\ge3$, and $\ball(c,\rho)$ denotes a closed ball. The order $\prec$ is the lexicographic order on $\R^d$, and the half-ball $\ball^+(c,\rho)$ consists of the points of $\ball(c,\rho)$ whose first coordinate is at least that of~$c$. Constants in the $O$-notation may depend on $d$. The local search does not depend on the dimension; \cref{sec:exp3d} reports experiments in $\R^3$.

\begin{theorem}\label{thm:sweepd}
Let $d\ge3$ be fixed and $\eps\le1/16$. There are pieces for which, for every $P\subset\R^d$, the sweep performs at most $\opt_\eps$ steps and outputs an $\eps$-WSPD of size less than $32^d\opt_\eps$, and at most $2048\,\opt_\eps$ if $d=3$, in $O(n^2\log n)$ time and $O(n^2)$ space. With the same pieces, the greedy sweep satisfies the same bounds, except that it runs in $O(n^3)$ time.
\end{theorem}

\subparagraph*{What carries over.}
\Cref{obs:hered,obs:local} hold in every dimension. The proof of \cref{thm:sweep} uses three properties of the pieces: they partition the ball and the half-ball, they have diameter at most $(1-2\eps)$ times the radius, and there are $8$ pieces in the ball and $5$ in the half-ball. The invariant and the charging also use that a point $x\succeq p$ of $\ball(p,\eps D)$ lies in $\ball^+(p,\eps D)$, which holds for the lexicographic order in any dimension. Hence, if the ball and the half-ball of radius $\rho$ are partitioned into $k$ and $h$ pieces, respectively, each of diameter at most $(1-2\eps)\rho$, the proofs of validity and of the approximation factor apply verbatim and give an $\eps$-WSPD of size at most $kh\,\opt_\eps$. \Cref{obs:anchors} does not use the geometry of the pieces, and it holds as stated: the number of steps is a lower bound on $\opt_\eps$. The same holds for the greedy sweep, whose cap of $40$ becomes $kh$.

\subparagraph*{Pieces.}
Let $\delta=7/8$, so that $\delta\le1-2\eps$ for $\eps\le1/16$. It suffices to partition $\ball(0,1)$ and $\ball^+(0,1)$ into pieces of diameter at most $\delta$, and to scale and translate them to $\ball(c,\rho)$ and $\ball^+(c,\rho)$. For fixed $d$, the pieces below are fixed sets, and a point is assigned to its piece in $O(1)$ time.

The \emph{grid pieces}: let $m_d=\lceil 8\sqrt d/7\rceil$, partition the bounding cube of $\ball(0,1)$ into the $(2m_d)^d$ cubes of side $1/m_d$ of a grid, and intersect them with the ball. The hyperplane through $0$ orthogonal to the first axis is a grid hyperplane, so the half-ball meets only $m_d(2m_d)^{d-1}$ of these cubes. Each cube has diameter $\sqrt d/m_d\le\delta$. Hence $k=(2m_d)^d$, $h=m_d(2m_d)^{d-1}$, and $kh=2^{2d-1}m_d^{2d}$. For $d=3$, we have $m_3=2$, with $64$ pieces in the ball and $32$ in the half-ball, and $kh=2048$. For $d=4$, we have $m_4=3$ and $kh=839808<32^4$. A point is assigned to its cube by rounding its coordinates down, and a coordinate equal to $1$ is assigned to the last cube.

The \emph{net pieces}: let $G$ be the set of points of the grid $\eta\mathbb{Z}^d$, with $\eta=\delta/(400\sqrt d)$, at distance at most $1+\delta/800$ from the origin, and let $N$ be a maximal subset of $G$ whose points are pairwise more than $399\delta/800$ apart, found greedily. Every point of $\ball(0,1)$ is within distance $\delta/800$ of $G$, and by maximality every point of $G$ is within distance $399\delta/800$ of $N$, so every point of $\ball(0,1)$ is within distance $\delta/2$ of $N$. Hence assigning every point to a nearest point of $N$ (ties broken by a fixed order of $N$) partitions the ball, and the half-ball, into at most $|N|$ pieces of diameter at most $\delta$. The balls of radius $399\delta/1600$ centered at $N$ are pairwise disjoint and contained in $\ball(0,1+401\delta/1600)$, so $|N|\le(1600/(399\delta)+401/399)^d=(15607/2793)^d<5.59^d$. Hence $kh<5.59^{2d}<31.25^d<32^d$, which is smaller than the factor of the grid pieces for $d\ge5$. The set $N$ depends only on $d$, and a point is assigned to its piece in $O(|N|)=O(1)$ time.

\begin{lemma}\label{lem:multd}
Let $\eps\le1/16$ and $d\ge 2$. Every pair is covered by the output dumbbells of at most $2(11/3)^{2d}$ steps of the sweep in $\R^d$, and by at most one output dumbbell of each step.
\end{lemma}
\begin{proof}
The proof of \cref{lem:multiplicity} applies verbatim, with balls in place of disks and volumes in $\R^{2d}$ in place of the products of areas. The products $\ball(p_j,s/2)\times\ball(q_j,s/2)\subset\R^{2d}$ are pairwise disjoint and contained in $\ball(x,r+s/2)\times\ball(y,r+s/2)$, so each of the two classes has at most $(1+2r/s)^{2d}\le(11/3)^{2d}$ steps.
\end{proof}

\begin{proof}[Proof of \cref{thm:sweepd}]
The size bounds follow from the discussion above, with the grid pieces for $d\le4$ and the net pieces for $d\ge5$. The proof of the running time in \cref{thm:sweep} applies verbatim, with \cref{lem:multd} in place of \cref{lem:multiplicity}, since a point is assigned to its piece in $O(1)$ time. For the greedy sweep, whose cap of $40$ becomes $kh$, the bounds on the number of steps and the size also follow from the discussion above, and the proof of its running time applies verbatim, with \cref{lem:greedymult} in $\R^d$, and gives $O(n^3)$ time and $O(n^2)$ space.
\end{proof}

As in the plane, the number of steps is at most $\opt_\eps=O(n/\eps^d)$~\cite{CK95}, so the sweep also runs in $O(n^2/\eps^d)$ time.

\subparagraph*{Dependence on $d$.}
The factor of the grid pieces grows as $d^{\Theta(d)}$, and that of the net pieces as $2^{O(d)}$. No choice of pieces can make this analysis subexponential in $d$. By the isodiametric inequality, a piece of diameter at most $(1-2\eps)\rho$ has at most a $((1-2\eps)/2)^d$ fraction of the volume of the ball, so $kh\ge 2^{2d-1}/(1-2\eps)^{2d}>2^{2d-1}$.

\section{Experiments}
\label{sec:experiments}

We address the following questions.
(Q1)~How close to optimal are the greedy sweep and the CK WSPD, and how does the gap depend on the input and $n$?
(Q2)~What are the running times?
(Q3)~Can local search close the remaining gap, and how does it compare with LNS?
(Q4)~How large are partition WSPDs converted from cover WSPDs, compared with those built directly?
We use $\eps=1/20$ (\cref{thm:sweep} needs $\eps\le1/16$). \Cref{sec:setup,sec:results} treat points in the plane and \cref{sec:exp3d} points in $\R^3$.

\subsection{Setup}
\label{sec:setup}

All algorithms are implemented in C++. Points have integer coordinates and $\eps$ is rational, so all separation tests are exact, and every output is checked by a separate brute-force verifier. We compare three algorithms. \emph{The sweep} is the greedy sweep of \cref{sec:greedy} followed by a pass that removes redundant dumbbells. \emph{CK} is the variant of the CK construction on a fair split tree~\cite{CK95} that accepts two nodes if their point sets are $\eps$-separated, and \emph{CK box} tests the bounding boxes of the two nodes instead, so it outputs at least as many dumbbells. The inputs (\cref{fig:inputs}) are six synthetic families (uniform, Gaussian, clustered, perturbed grid, logarithmic spiral, and annulus) and the non-uniform instances of the CG:SHOP 2019 challenge~\cite{cgshop2019}, for $n$ from $200$ to $50000$. Code, generators, and raw results are available at \url{https://github.com/gfonsecabr/small-wspd-experiments}. \Cref{tab:lstime} gives the local search over time, \cref{tab:exp} the mean size divided by the reference (\cref{sec:results}) for every input and $n\in\{1000,10000,50000\}$, and \cref{tab:partconv} the same for the partition WSPDs of Q4 and $n\in\{1000,10000\}$. The rest of this subsection gives the details.

\begin{figure}[htb]
\centering
\includegraphics[width=\linewidth]{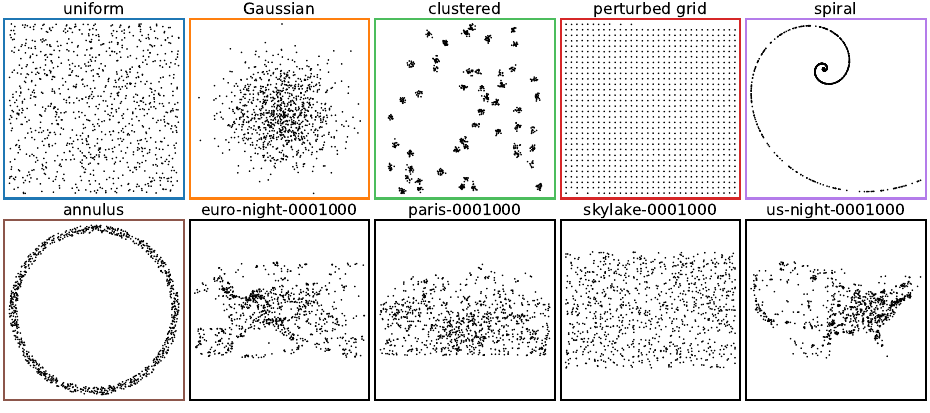}
\caption{Synthetic and CG:SHOP instances for $n=1000$.}
\label{fig:inputs}
\end{figure}

\subparagraph*{Measurements.}
All experiments ran on an AMD Ryzen~9 9900X (12 cores, 128\,GB of RAM, Fedora Linux~44, GCC~16.2 with \texttt{-O2}), with up to $11$ single-threaded runs in parallel and no other load. Times are wall-clock. Memory is the peak resident set size of a run, which includes the output and the verifier. All methods, including CK, store every dumbbell as two lists of points, so the times of CK include listing these points, which the time bounds of \cref{sec:partition} do not. In \cref{fig:time}, the time to $1.001\,\opt$ is the mean over the inputs with proven $\opt$, and where CPLEX reached its time limit of $1800$\,s on some input, the mean time of CPLEX is a lower bound.

\subparagraph*{Implementation.}
The implementation uses the Pangolin geometry library~\cite{pgl}. With $\eps=a/b$, the separation test is $b^2\diam(X)^2\le a^2\dist(X,Y)^2$, and the verifier checks coverage and exact separation.

\emph{Sweep.} The final pass that removes redundant dumbbells considers them in the order of increasing $|X|\cdot|Y|$. For each $p$, the sweep processes the points $q\succ p$ by increasing distance from~$p$. It stores one byte per pair, so memory is quadratic in $n$ (\cref{fig:time}). When a step falls back to the fixed collection, it assigns the points to the pieces in floating point, which matters only for points on a boundary between pieces, and the output is verified like every other.

\emph{CK.} The recursive construction on a fair split tree~\cite{CK95}, which splits the longer side of the bounding box of a node at its midpoint. When two nodes are not accepted, the recursion splits the one whose bounding box has the larger diameter. CK box accepts two nodes if the diameters of their boxes are at most $\eps$ times the distance between the boxes, and CK tests the point sets exactly when this test fails. These are the two variants of \cref{thm:fairsplit}. The original construction~\cite{CK95} tests balls that contain the boxes and splits the node whose box has the longer side (\cref{sec:partition}). Both variants are partition WSPDs, so the redundancy pass would not change them.

\emph{Exact optimum.} An $\eps$-separated dumbbell $\{X,Y\}$ with $\dist(X,Y)=\|pq\|$, $p\in X$, and $q\in Y$ corresponds to a clique containing $p$ and $q$ in an auxiliary graph on the points of $\ball(p,\eps\|pq\|)\cup\ball(q,\eps\|pq\|)$. The edges of this graph encode the pairwise conditions $\|xx'\|\le\eps\,\|pq\|$, $\|yy'\|\le\eps\,\|pq\|$, and $\|xy\|\ge\|pq\|$. We enumerate all maximal cliques (Bron--Kerbosch~\cite{BK73}) for every choice of $p$ and $q$, and solve the resulting set-cover problem by integer programming with CPLEX~22.2. This is exact because some minimum WSPD consists of maximal dumbbells, and every maximal dumbbell arises from a maximal clique. Where CPLEX proves optimality for $n=1000$ ($10$ of the $22$ instances), the set-cover instance has a median of $6.6\cdot10^5$ maximal dumbbells. For the spirals with $n\ge800$, we first apply exact set-cover reductions (forced columns, row and column dominance) until none applies, and then run CPLEX on one connected component at a time, with three threads and the sweep followed by reshape as a warm start. For $n=1000$, $\opt$ is thus known on $13$ instances.

\emph{Lower bound.} The lower bounds come from separate runs of the sweep, which write the anchors to a file, followed by the search of \cref{sec:lb}, a separate program, with up to $9$ runs in parallel for $n\le10000$ and $4$ for larger $n$. The plateau moves use one fixed random seed. In \cref{fig:lb}, panel~(a) shows the mean over the instances with proven $\opt$ (all of them for $n\le500$), panel~(b) the mean over all instances, and panel~(c) the mean running time of these runs, including the sweep. We computed the independence number $\alpha$ of $G_\eps$ for every instance with $n\le500$ by integer programming with CPLEX: one $0$/$1$ variable per pair, whose sum is maximized, with at most one chosen pair in every maximal dumbbell, since two pairs are compatible if and only if some maximal dumbbell covers both. CPLEX proves optimality on all $88$ instances, in at most $2$\,min each.

\emph{Local search.} Every run uses one fixed random seed. For LNS, the \emph{scale} of a dumbbell $\{X,Y\}$ is $\dist(X,Y)$. The \emph{neighborhood} of a dumbbell $T=\{A,B\}$ with scale $\ell_T$ consists of the dumbbells $Z$ that contain a point within distance $2\eps\ell_T$ of a point of $T$, whose scale is within a factor $(1+2\eps)^2$ of $\ell_T$, and whose sides, in one of the two orientations, have centroids within distance $2\eps\ell_T$ of the centroids of $A$ and $B$. Neighbors are ranked by the larger of these two centroid distances, and LNS removes $T$ and its $k$ nearest neighbors. The candidate dumbbells for a pair of $U$ are the first $300$ found by the enumeration used to compute $\opt$ (above), and CPLEX has a limit of $2$\,s per call. We use $(k,u)=(16,1000)$ for $n\le1000$ and $(32,400)$ for larger $n$. In \cref{fig:lstotal}, a curve begins at the mean time to build the start, over the instances of an input, and is horizontal while reshape is set up (relabeling the points and building its tables), a time that we measured in separate runs. It then shows the mean size for the same reshape time. Within its budget, LNS makes about $2\cdot10^5$ moves for $n=10000$ and between $8\cdot10^4$ and $3\cdot10^5$ for $n=50000$, on WSPDs of $3\cdot10^6$ to $4\cdot10^7$ dumbbells, so for large $n$ it tries only a small fraction of the dumbbells. The parameters of LNS are the best ones among those we tested on the local-search samples (below) and on three instances with $n=3000$: neighborhoods of radius $1.5\,\eps\ell_T$ to $5\,\eps\ell_T$, $4$ to $64$ neighbors, and bounds on $|U|$ from $50$ to $4000$ and no bound, not in all combinations. The limit of CPLEX ($0.5$ to $10$\,s) and the number of dumbbells enumerated per pair ($50$ to $2000$) changed the sizes by at most $0.3\%$.

\subparagraph*{Inputs.}
\emph{Synthetic inputs} (\cref{fig:inputs}) are generated at scale $M=10^6$ around the center $(M/2,M/2)$ and rounded to integer coordinates, with three seeds each: \emph{uniform} in $[0,M]^2$; \emph{Gaussian} (standard deviation $M/8$); \emph{clustered} ($\max\{2,n/20\}$ uniform centers, each point Gaussian around a random center with standard deviation $M/100$); \emph{perturbed grid} (a grid with $\lceil\sqrt n\rceil$ columns, each point moved by a Gaussian of standard deviation $2\%$ of the spacing); \emph{logarithmic spiral} (the point $2000\,e^t(\cos3t,\sin3t)$ relative to the center, with $t$ uniform in $[0,6]$); and \emph{annulus} (uniform angle, radius uniform in $[1,1.1]\cdot M/4$). Duplicate points are discarded and replaced. The \emph{CG:SHOP 2019} instances~\cite{cgshop2019} come from the 2019 Computational Geometry Challenge. We use all of its non-uniform instances whose size is one of the values of $n$ below, after removing duplicates. We use $n$ from $200$ to $1000$ in steps of $100$, from $2000$ to $10000$ in steps of $1000$, and from $20000$ to $50000$ in steps of $10000$ (one seed). The perturbed grid is filled row by row.

For every $n\le10000$, there are $18$ synthetic and $4$ CG:SHOP instances, and for every $n\ge20000$, $6$ and $4$. The local-search samples are the first seed of every generator and the four CG:SHOP instances for $n=1000$, and the first uniform and clustered instances and one CG:SHOP instance (euro-night) for $n=10000$ and $n=50000$. \Cref{fig:ls} and \cref{tab:lstime} show reshape and LNS on these samples, and \cref{fig:ls} also on all instances with $n=200$. For $n=50000$, both ran for $5000$\,s.

\subsection{Results}
\label{sec:results}

\subparagraph*{Reference value.}
We compute $\opt$ as described in \cref{sec:setup}. Within $1800$\,s per instance, CPLEX proves optimality on all instances with $n\le500$ and on some with $600\le n\le1000$, including the spirals with $n\ge800$.
Otherwise we compare against the lower bound $\lb$ of \cref{sec:lb}, so some values of $n$ mix both references. Where $\opt$ is known (\cref{fig:lb}), $\lb$ is on average $0.999$, $0.997$, and $0.994$ times $\opt$ for $n=200$, $500$, and $1000$, at least $0.989\,\opt$ on every instance, and equal to $\opt$ on $7$ instances. The number of steps $\lb_0$ is on average $0.98$, $0.97$, and $0.94$ times $\opt$, and at least $0.90\,\opt$, and $\lb_1$ is in between. For $n\le500$, we computed exactly the independence number $\alpha$ of $G_\eps$, the best bound that independent sets can give: it is at least $0.992\,\opt$ and on average $0.999\,\opt$. Of the gap between $\lb_0$ and $\alpha$, the swaps close $76\%$ and the plateau moves another $22\%$, and $\lb=\alpha$ on $31$ of these $88$ instances. Ratios to $\lb$ overestimate ratios to $\opt$, probably more so for larger $n$: $\lb/\lb_0$ grows from $1.05$ for $n=1000$ to $1.10$ and $1.12$ for $n=10000$ and $50000$.

\begin{figure}[htb]
\centering
\includegraphics[width=\linewidth]{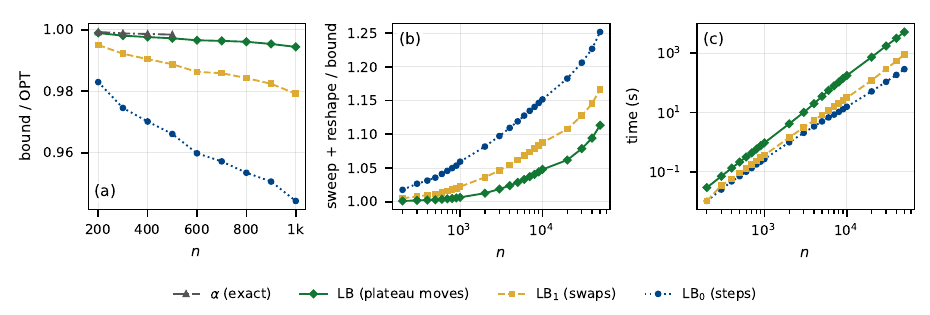}
\caption{The lower bounds of \cref{sec:lb}: means over the instances.}
\label{fig:lb}
\end{figure}

\subparagraph*{Quality (Q1).}
\Cref{fig:summary} summarizes the results, and \cref{fig:ratio,tab:exp} show them per input. For $n=500$, the sweep is on average $3\%$ above $\opt$ (at most $9\%$), and CK box and CK are $21\%$ and $17\%$ above (at most $37\%$ and $27\%$). For $n=10000$, the sweep is on average at $1.22\,\lb$ and CK at $1.55\,\lb$, and for $n=50000$ at $1.38\,\lb$ and $1.76\,\lb$. Part of this growth may come from the lower bound, whose ratio to $\opt$ decreases with $n$ where $\opt$ is known (\cref{fig:lb}). The sweep is smallest, followed by CK and CK box, on every input and for every $n$, except on the perturbed grid, where the three methods are equal for $n=200$ and CK is smaller than the sweep for $n=1000$ (by less than $1\%$) and for $n=4000$ (by $4\%$). After the merge pass of Q4, CK is still larger than the sweep on average for every $n$ ($1.06\,\opt$ against $1.03\,\opt$ for $n=500$, and $1.46\,\lb$ against $1.22\,\lb$ for $n=10000$). \Cref{fig:perpoint} gives the number of dumbbells per point of the sweep followed by reshape. It grows with $n$ on every input except the spiral, where it stays near $22$, and for $n=50000$ it ranges from $23$ on the spiral to $981$ on the perturbed grid.

\begin{figure}[htb]
\centering
\includegraphics[width=\linewidth]{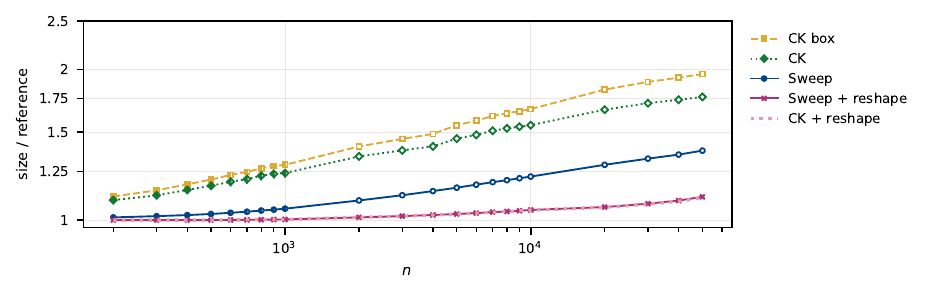}
\caption{Means. Open (half-filled) markers: the reference is $\lb$ on all (some) instances.}
\label{fig:summary}
\end{figure}
\begin{figure}[htb]
\centering
\includegraphics[width=\linewidth]{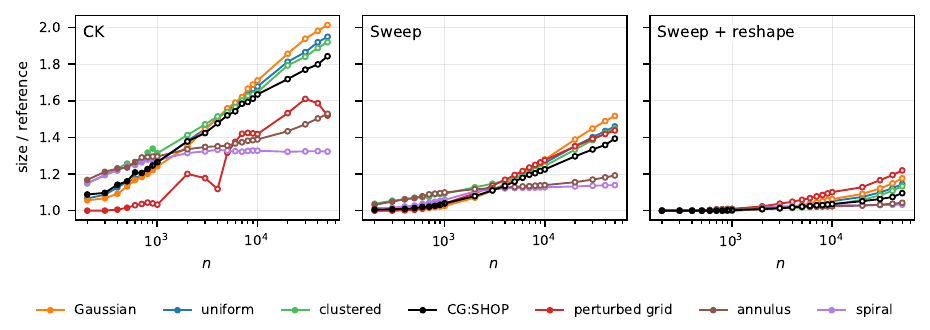}
\caption{Mean size divided by the reference, per input. Markers as in \cref{fig:summary}.}
\label{fig:ratio}
\end{figure}
\begin{figure}[htb]
\centering
\includegraphics[width=\linewidth]{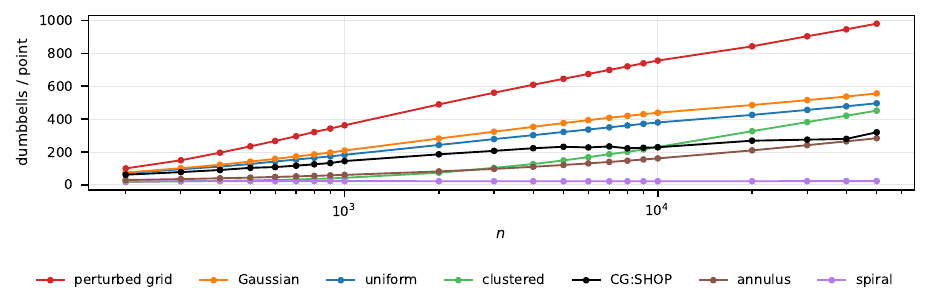}
\caption{Sweep + reshape: mean number of dumbbells per point, per input.}
\label{fig:perpoint}
\end{figure}

\subparagraph*{Running times (Q2).}
\Cref{fig:time} reports the running times and the memory. On average, CK box and CK take $0.5$\,s for $n=10000$ (at most $1.2$\,s) and $3.7$\,s and $4.1$\,s for $n=50000$, and the sweep takes $0.3$\,s for $n=1000$, $15$\,s for $n=10000$, and $5$\,min for $n=50000$. For $n=50000$, CK, the sweep, and reshape use up to $8$, $15$, and $22$\,GB. Computing $\opt$, including the enumeration of the maximal dumbbells, takes $0.07$\,s for $n=200$ and $15$\,s for $n=1000$ (median where CPLEX proves optimality). After the sweep, computing $\lb_1$ takes on average $16$\,s for $n=10000$ and $10$\,min for $n=50000$, and computing $\lb$ takes $2.7$\,min and $82$\,min (\cref{fig:lb}).

\begin{figure}[htb]
\centering
\includegraphics[width=\linewidth]{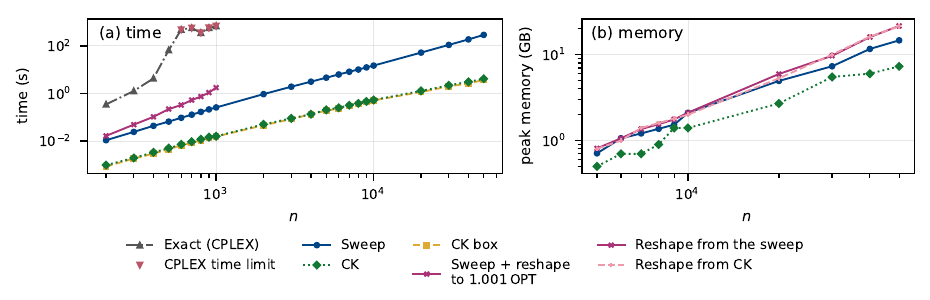}
\caption{Mean running time and largest peak memory.}
\label{fig:time}
\end{figure}

\begin{figure}[htb]
\centering
\includegraphics[width=\linewidth]{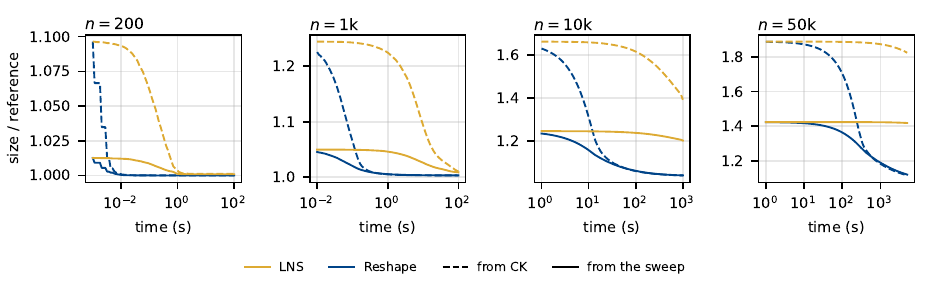}
\caption{Local search: mean size / reference over time, on all instances for $n=200$.}
\label{fig:ls}
\end{figure}

\subparagraph*{Local search (Q3).}
We ran reshape (\cref{sec:reshape}) from CK and from the sweep on every instance, for $100$\,s for $n\le1000$, $1000$\,s for $2000\le n\le10000$, and $5000$\,s above. LNS ran with the same budgets, on every instance for $n\le1000$ and on samples of $3$ instances for $n=10000$ and $n=50000$ (\cref{fig:ls,tab:lstime}).
Reshape ends at most as large as LNS from the same start on every instance, and smaller on average for every $n$. It reaches the final size of LNS within $1\%$ of the budget, while LNS improves CK from $1.66$ to $1.39$ times the reference for $n=10000$ and from $1.89$ to $1.82$ for $n=50000$, where it tries only a small fraction of the dumbbells within its budget (\cref{sec:setup}). On the $164$ instances with $n\le1000$ where $\opt$ is known, the sweep followed by reshape is on average $0.002\%$ and at most $0.02\%$ above $\opt$, and it reaches $\opt$ on $101$ of them (all $22$ for $n=200$, $13$ of $22$ for $n=500$, and $1$ of $13$ for $n=1000$), against $0.23\%$, $1.5\%$, and $18$ for LNS from the sweep. It is on average at $1.05\,\lb$ for $n=10000$ (at most $1.10\,\lb$) and $1.11\,\lb$ for $n=50000$ (at most $1.22\,\lb$), against $1.22\,\lb$ and $1.38\,\lb$ for the sweep alone (right column of \cref{fig:ratio}). For the same total time, including the time to build the start and to set up reshape, the best results are obtained by CK followed by reshape (\cref{fig:lstotal}), although CK is on average larger than the output of the sweep. On every instance with $n\ge20000$ and on $378$ of the $396$ instances with $n\le10000$, it is smaller than the output of the sweep before the sweep has finished; this takes on average $89\%$ of the running time of the sweep for $n=10000$ and $76\%$ for $n=50000$. At every later time and on every instance, the sweep followed by reshape is at most $0.73\%$ smaller, and never smaller for $n\ge20000$. However, the approximation factor that we prove for CK followed by reshape is larger: since reshape never increases the size, it is $576$ for $\eps\le1/16$ (\cref{sec:partition}), against $40$ for the sweep followed by reshape (\cref{thm:greedy}). For the same reshape time, the two starts end within $0.42\%$ of each other for $n\le10000$. For $n\ge20000$, where reshape has not converged after $5000$\,s, the run from CK ends smaller on $39$ of the $40$ instances, by $0.4\%$ on average for $n=50000$ and by at most $1.2\%$. Since the two starts end close to each other, the sweep matters mainly for its lower bound and its approximation factor when reshape follows.

\begin{figure}[htb]
\centering
\includegraphics[width=\linewidth]{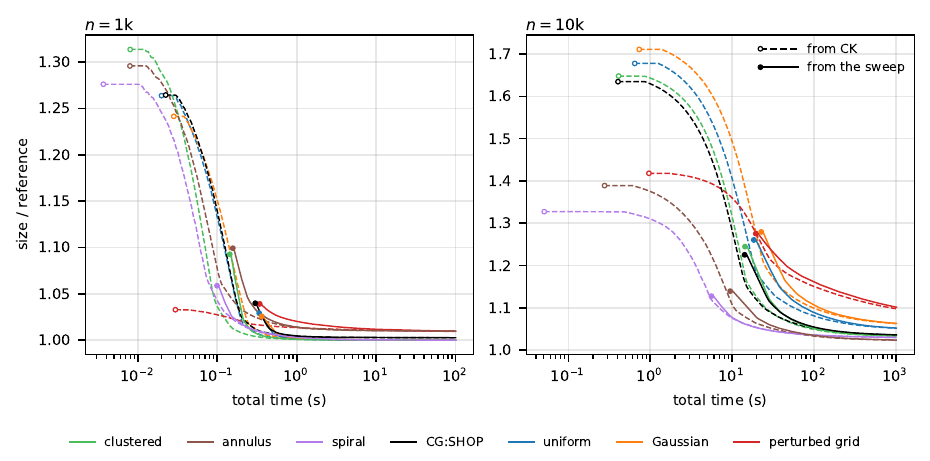}
\caption{Mean size / reference per input over the total time: build of the start, then reshape.}
\label{fig:lstotal}
\end{figure}

\subparagraph*{Partition WSPDs (Q4).}
We convert a cover WSPD into a partition WSPD by processing its dumbbells one by one. For the current dumbbell $\{X,Y\}$, let $U$ be the set of its pairs that no earlier dumbbell covers. Two points of $X$ are grouped if they have the same partners in $Y$ among the pairs of $U$, and each group $X'\subseteq X$ with partners $Y'\subseteq Y$ gives the dumbbell $\{X',Y'\}$. The same is done with the roles of $X$ and $Y$ exchanged, and the split into fewer dumbbells is kept. These dumbbells cover the pairs of $U$ exactly once, and they are $\eps$-separated by \cref{obs:hered}. The order comes from a priority queue that holds, for every dumbbell, the number of dumbbells into which its uncovered pairs split: the next dumbbell has the smallest number, ties broken by size, largest first. The numbers are updated lazily.

A final \emph{merge pass} replaces two dumbbells $\{S,O_1\}$ and $\{S,O_2\}$ with a common side by $\{S,O_1\cup O_2\}$ if it is $\eps$-separated. The result is still a partition WSPD. The pass tries at most $256$ dumbbells per side and is repeated until no merge occurs, at most $20$ times. It applies to every partition WSPD, also to CK and the quadtree.

The conversion has no size guarantee (\cref{fig:partconv,tab:partconv}). For the sweep followed by reshape, the partition WSPD is on average $1.015$, $1.13$, and $1.18$ times the reference for $n=1000$, $10000$, and $20000$ (at most $1.40$, on the perturbed grid with $n=20000$), against $1.003$, $1.05$, and $1.06$ for the cover WSPD, and the conversion takes at most $54$\,s for $n=10000$ and $142$\,s for $n=20000$. CK, which is a partition WSPD, is at $1.24$, $1.55$, and $1.66$, and at $1.11$, $1.46$, and $1.59$ after the merge pass. It is never smaller than the converted WSPD, also after the merge pass. Converting the optimal WSPDs with $n=500$ gives partition WSPDs of size $1.004\,\opt_\eps$ on average and at most $1.0125\,\opt_\eps$, which bounds the partition optimum on these instances. The quadtree construction of \cref{sec:partition}, with exact tests, is at $1.71$, $2.45$, and $2.78$ times the reference, and at $1.15$, $1.98$, and $2.36$ after the merge pass. Despite its smaller proven factor, it is larger than CK on average for every $n$, with or without merging.

\begin{figure}[htb]
\centering
\includegraphics[width=\linewidth]{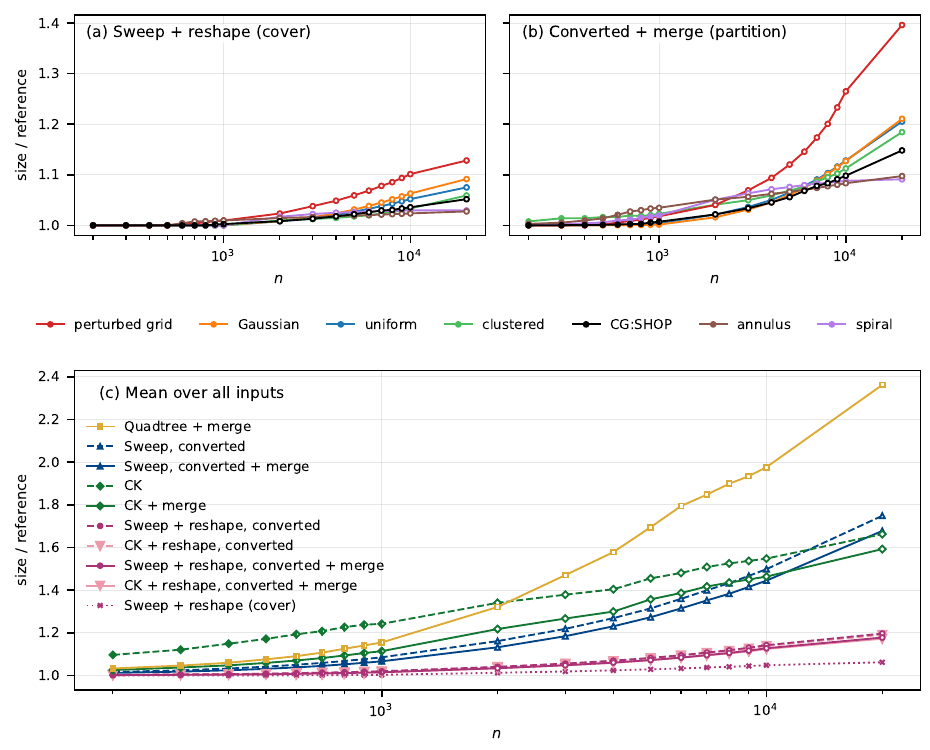}
\caption{Partition WSPDs: mean size / reference, (a,\,b)~per input. Markers as in \cref{fig:summary}.}
\label{fig:partconv}
\end{figure}

\clearpage
{
\let\origtable\table\renewcommand{\table}[1][]{\origtable[p]}
\begin{table}[htbp]
\centering\footnotesize\setlength{\tabcolsep}{3pt}
\caption{Local search over time: mean size divided by the reference. \emph{Build time}: median time to build the start, not included in the times. Bold: best value for each $n$ and time over both starts.}
\label{tab:lstime}
\begin{tabular}{@{}l@{\hspace{8pt}}rrr@{\hspace{8pt}}rrr@{\hspace{8pt}}rrr@{}}
\toprule
$n$ & \multicolumn{3}{c}{1k} & \multicolumn{3}{c}{10k} & \multicolumn{3}{c}{50k} \\
\cmidrule(lr){2-4} \cmidrule(lr){5-7} \cmidrule(lr){8-10}
Time & 1\,s & 10\,s & 100\,s & 10\,s & 100\,s & 1000\,s & 50\,s & 500\,s & 5000\,s \\
\midrule
\multicolumn{10}{@{}l}{\emph{Start: CK}} \\
Build time & \multicolumn{3}{c}{0.017\,s} & \multicolumn{3}{c}{0.41\,s} & \multicolumn{3}{c}{5.7\,s} \\
Start & 1.245 & 1.245 & 1.245 & 1.662 & 1.662 & 1.662 & 1.889 & 1.889 & 1.889 \\
Reshape & \textbf{1.005} & \textbf{1.003} & \textbf{1.003} & 1.312 & \textbf{1.061} & \textbf{1.041} & 1.801 & 1.240 & \textbf{1.116} \\
LNS & 1.225 & 1.093 & 1.011 & 1.658 & 1.616 & 1.392 & 1.888 & 1.881 & 1.817 \\
\midrule
\multicolumn{10}{@{}l}{\emph{Start: Sweep}} \\
Build time & \multicolumn{3}{c}{0.31\,s} & \multicolumn{3}{c}{15\,s} & \multicolumn{3}{c}{354\,s} \\
Start & 1.050 & 1.050 & 1.050 & 1.247 & 1.247 & 1.247 & 1.424 & 1.424 & 1.424 \\
Reshape & \textbf{1.005} & \textbf{1.003} & \textbf{1.003} & \textbf{1.158} & 1.063 & \textbf{1.041} & \textbf{1.392} & \textbf{1.235} & 1.121 \\
LNS & 1.046 & 1.026 & 1.009 & 1.246 & 1.238 & 1.203 & 1.424 & 1.424 & 1.418 \\
\bottomrule
\end{tabular}
\end{table}

\begin{table}[htbp]
\centering\footnotesize\setlength{\tabcolsep}{4pt}
\caption{Mean size divided by the reference, per input: synthetic inputs (three seeds for $n\le10000$, one for $n=50000$) and CG:SHOP 2019 families. $\dagger$: the reference is the lower bound. Bold: smallest construction of the row.}\label{tab:exp}
\begin{tabular}{@{}lrccccc@{}}
\toprule
Input & $n$ & \shortstack{\strut CK box} & \shortstack{\strut CK} & \shortstack{\strut Sweep} & \shortstack{\strut CK\\\strut + reshape} & \shortstack{\strut Sweep\\\strut + reshape} \\
\midrule
uniform & 1k\phantom{$^\dagger$} & 1.30 & 1.26 & \textbf{1.03} & 1.00 & 1.00 \\
 & 10k$^\dagger$ & 1.84 & 1.68 & \textbf{1.26} & 1.05 & 1.05 \\
 & 50k$^\dagger$ & 2.21 & 1.95 & \textbf{1.46} & 1.14 & 1.15 \\
\addlinespace[2pt]
Gaussian & 1k\phantom{$^\dagger$} & 1.27 & 1.24 & \textbf{1.03} & 1.00 & 1.00 \\
 & 10k$^\dagger$ & 1.88 & 1.71 & \textbf{1.28} & 1.06 & 1.06 \\
 & 50k$^\dagger$ & 2.30 & 2.01 & \textbf{1.52} & 1.17 & 1.18 \\
\addlinespace[2pt]
clustered & 1k\phantom{$^\dagger$} & 1.44 & 1.31 & \textbf{1.09} & 1.00 & 1.00 \\
 & 10k$^\dagger$ & 1.84 & 1.65 & \textbf{1.24} & 1.03 & 1.03 \\
 & 50k$^\dagger$ & 2.18 & 1.92 & \textbf{1.45} & 1.13 & 1.13 \\
\addlinespace[2pt]
perturbed grid & 1k$^\dagger$ & \textbf{1.03} & \textbf{1.03} & 1.04 & 1.01 & 1.01 \\
 & 10k$^\dagger$ & 1.45 & 1.42 & \textbf{1.28} & 1.10 & 1.10 \\
 & 50k$^\dagger$ & 1.56 & 1.52 & \textbf{1.44} & 1.21 & 1.22 \\
\addlinespace[2pt]
spiral & 1k\phantom{$^\dagger$} & 1.28 & 1.28 & \textbf{1.06} & 1.00 & 1.00 \\
 & 10k$^\dagger$ & 1.33 & 1.33 & \textbf{1.13} & 1.03 & 1.03 \\
 & 50k$^\dagger$ & 1.32 & 1.32 & \textbf{1.14} & 1.03 & 1.03 \\
\addlinespace[2pt]
annulus & 1k$^\dagger$ & 1.38 & 1.30 & \textbf{1.10} & 1.01 & 1.01 \\
 & 10k$^\dagger$ & 1.49 & 1.39 & \textbf{1.14} & 1.02 & 1.02 \\
 & 50k$^\dagger$ & 1.66 & 1.53 & \textbf{1.19} & 1.04 & 1.04 \\
\midrule
euro-night & 1k\phantom{$^\dagger$} & 1.30 & 1.25 & \textbf{1.04} & 1.00 & 1.00 \\
 & 10k$^\dagger$ & 1.80 & 1.64 & \textbf{1.23} & 1.04 & 1.04 \\
 & 50k$^\dagger$ & 2.02 & 1.80 & \textbf{1.36} & 1.08 & 1.09 \\
\addlinespace[2pt]
us-night & 1k$^\dagger$ & 1.34 & 1.29 & \textbf{1.04} & 1.00 & 1.00 \\
 & 10k$^\dagger$ & 1.77 & 1.61 & \textbf{1.21} & 1.03 & 1.03 \\
 & 50k$^\dagger$ & 1.97 & 1.76 & \textbf{1.34} & 1.07 & 1.08 \\
\addlinespace[2pt]
paris & 1k$^\dagger$ & 1.32 & 1.27 & \textbf{1.04} & 1.00 & 1.00 \\
\addlinespace[2pt]
skylake & 1k$^\dagger$ & 1.28 & 1.24 & \textbf{1.04} & 1.00 & 1.00 \\
\addlinespace[2pt]
jupiter & 10k$^\dagger$ & 1.81 & 1.65 & \textbf{1.22} & 1.03 & 1.03 \\
\addlinespace[2pt]
world & 10k$^\dagger$ & 1.81 & 1.64 & \textbf{1.24} & 1.04 & 1.04 \\
\addlinespace[2pt]
protein & 50k$^\dagger$ & 2.15 & 1.88 & \textbf{1.43} & 1.11 & 1.11 \\
\addlinespace[2pt]
sun & 50k$^\dagger$ & 2.21 & 1.94 & \textbf{1.44} & 1.10 & 1.11 \\
\bottomrule
\end{tabular}
\end{table}

\begin{table}[htbp]
\centering\footnotesize\setlength{\tabcolsep}{4pt}
\caption{Partition WSPDs: mean size divided by the reference ($\dagger$: the lower bound) per input and $n$. First column: the cover WSPD of the sweep followed by reshape. The others are partition WSPDs: its conversion and CK, each without and with the merge pass, and the deterministic quadtree with exact separation tests and the merge pass. Bold: smallest partition WSPD of the row.}
\label{tab:partconv}
\begin{tabular}{@{}lrcccccc@{}}
\toprule
Input & $n$ & \shortstack{\strut Sweep\\\strut + reshape} & \shortstack{\strut Converted} & \shortstack{\strut Converted\\\strut + merge} & \shortstack{\strut CK} & \shortstack{\strut CK\\\strut + merge} & \shortstack{\strut Quadtree\\\strut + merge} \\
\midrule
uniform & 1k\phantom{$^\dagger$} & 1.00 & 1.01 & \textbf{1.00} & 1.26 & 1.12 & 1.17 \\
 & 10k$^\dagger$ & 1.05 & 1.14 & \textbf{1.13} & 1.68 & 1.64 & 2.54 \\
\addlinespace[2pt]
Gaussian & 1k\phantom{$^\dagger$} & 1.00 & \textbf{1.00} & \textbf{1.00} & 1.24 & 1.12 & 1.19 \\
 & 10k$^\dagger$ & 1.06 & 1.14 & \textbf{1.13} & 1.71 & 1.68 & 2.61 \\
\addlinespace[2pt]
clustered & 1k\phantom{$^\dagger$} & 1.00 & 1.04 & \textbf{1.02} & 1.31 & 1.14 & 1.19 \\
 & 10k$^\dagger$ & 1.03 & 1.13 & \textbf{1.11} & 1.65 & 1.55 & 2.19 \\
\addlinespace[2pt]
perturbed grid & 1k$^\dagger$ & 1.01 & \textbf{1.02} & \textbf{1.02} & 1.03 & \textbf{1.02} & 1.05 \\
 & 10k$^\dagger$ & 1.10 & 1.28 & \textbf{1.26} & 1.42 & 1.41 & 1.80 \\
\addlinespace[2pt]
spiral & 1k\phantom{$^\dagger$} & 1.00 & 1.03 & \textbf{1.02} & 1.28 & 1.16 & 1.17 \\
 & 10k$^\dagger$ & 1.03 & 1.10 & \textbf{1.09} & 1.33 & 1.22 & 1.23 \\
\addlinespace[2pt]
annulus & 1k$^\dagger$ & 1.01 & 1.05 & \textbf{1.03} & 1.30 & 1.13 & 1.16 \\
 & 10k$^\dagger$ & 1.02 & 1.10 & \textbf{1.08} & 1.39 & 1.20 & 1.26 \\
\addlinespace[2pt]
euro-night & 1k\phantom{$^\dagger$} & 1.00 & 1.01 & \textbf{1.00} & 1.25 & 1.10 & 1.15 \\
 & 10k$^\dagger$ & 1.04 & 1.11 & \textbf{1.10} & 1.64 & 1.54 & 2.11 \\
\addlinespace[2pt]
us-night & 1k$^\dagger$ & 1.00 & \textbf{1.01} & \textbf{1.01} & 1.29 & 1.11 & 1.15 \\
 & 10k$^\dagger$ & 1.03 & 1.10 & \textbf{1.09} & 1.61 & 1.51 & 2.02 \\
\addlinespace[2pt]
paris & 1k$^\dagger$ & 1.00 & \textbf{1.01} & \textbf{1.01} & 1.27 & 1.12 & 1.16 \\
\addlinespace[2pt]
skylake & 1k$^\dagger$ & 1.00 & \textbf{1.01} & \textbf{1.01} & 1.24 & 1.11 & 1.16 \\
\addlinespace[2pt]
jupiter & 10k$^\dagger$ & 1.03 & 1.11 & \textbf{1.09} & 1.65 & 1.56 & 2.27 \\
\addlinespace[2pt]
world & 10k$^\dagger$ & 1.04 & 1.13 & \textbf{1.11} & 1.64 & 1.54 & 2.20 \\
\bottomrule
\end{tabular}
\end{table}

}
\clearpage

\subsection{Experiments in Three Dimensions}
\label{sec:exp3d}

We ran the experiments of \cref{sec:setup,sec:results} on points in $\R^3$, with $\eps=1/20$: the sweep, CK box, CK, reshape and LNS, the conversion into partition WSPDs, and the computations of $\opt$, $\lb$, and $\alpha$. \Cref{tab:exp3d} gives the results per input for $n\in\{1000,10000\}$.

\subparagraph*{Implementation.}
The experiments in $\R^3$ use a separate C++ implementation, with its own kd-tree. As in the plane, points have integer coordinates, all separation tests are exact, and every output is checked by a brute-force verifier for coverage and exact separation, and, for partition WSPDs, for covering every pair exactly once. The greedy sweep uses the grid pieces of \cref{sec:higher}: $64$ cubes of side $\eps D/2$ in the ball and the $32$ of them that meet the half-ball. A step falls back to the fixed collection only if it has fewer nonempty products than the dumbbells grown in the step, which happened in $117$ steps over all instances, at most a fraction $10^{-5}$ of the steps of any instance. CK splits the longest side of the bounding box of a node at its midpoint, and reshape and LNS relabel the points in Morton order. On planar instances embedded in the plane $z=0$, CK box and CK give exactly the sizes of the planar implementation, and the sweep is within $0.15\%$ of it: ties are broken by the order in which the range queries report the points. The partition WSPDs are the conversion of Q4 with the merge pass, and the octree version of the quadtree construction of \cref{sec:partition}, with exact tests. $\opt$, $\lb$ (with $30$ plateau passes), $\alpha$, and LNS are computed as in the plane (\cref{sec:lb,sec:setup}). Reshape runs from CK and from the sweep on every instance, for $100$\,s for $n\le1000$ and $1000$\,s for larger $n$, and LNS from both starts, with the same budgets, on all instances with $n=200$, on all instances with $n=1000$, and on three instances with $n=10000$ (uniform, clustered, and the bunny). All runs are single-threaded, with up to $11$ runs in parallel, on the machine of \cref{sec:setup}; the times of the sweep, of CK box, and of CK come from separate runs of each.

\subparagraph*{Inputs.}
The \emph{synthetic inputs} (\cref{fig:inputs3d}) are generated in the cube $[0,M]^3$, with $M=10^6$, and rounded to integer coordinates: \emph{uniform} in the cube; \emph{Gaussian} (around the center of the cube, standard deviation $M/8$ in each coordinate); \emph{clustered} ($\max\{2,n/20\}$ uniform centers, each point Gaussian around a random center with standard deviation $M/100$); \emph{perturbed grid} (a grid with $\lceil n^{1/3}\rceil$ points per side, filled layer by layer, each point moved by a Gaussian of standard deviation $2\%$ of the spacing); \emph{conical spiral} (the point $2000\,e^t(\cos3t,\sin3t,1)$ relative to the center of the cube, with $t$ uniform in $[0,6]$); and \emph{spherical shell} (around the center of the cube, uniform direction, radius uniform in $[1,1.1]\cdot M/4$). Duplicate points are discarded and replaced. The \emph{scans} are four meshes of the Stanford 3D Scanning Repository~\cite{Stanford3D}: the bunny, the armadillo, the XYZ RGB dragon, and the happy Buddha, with $35947$, $49990$, $125066$, and $49251$ vertices. Each is scaled so that the longest side of its bounding box is $10^6$ and rounded, and we use a random sample of $n$ vertices. We use the values of $n$ of \cref{sec:setup} up to $10000$: from $200$ to $1000$ in steps of $100$ and from $2000$ to $10000$ in steps of $1000$. For each $n$, there is one instance per generator and one sample of each scan, so $10$ instances, and $180$ in all.

\begin{figure}[htb]
\centering
\includegraphics[width=\linewidth]{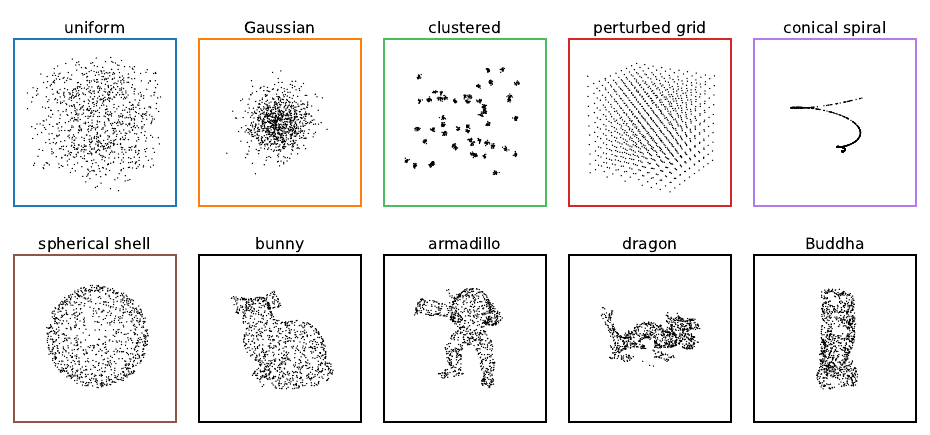}
\caption{Synthetic instances and scans in $\R^3$ for $n=1000$.}
\label{fig:inputs3d}
\end{figure}

\subparagraph*{Reference value.}
CPLEX proves optimality on all $90$ instances with $n\le1000$. Computing $\opt$ takes a median of $3.2$\,s for $n=1000$, but up to $55$\,min on the conical spirals, where the enumeration of the maximal dumbbells dominates. For larger $n$, the reference is $\lb$. Where $\opt$ is known (\cref{fig:lb3d}), $\lb$ is on average $0.999$, $0.998$, and $0.998$ times $\opt$ for $n=200$, $500$, and $1000$, at least $0.986\,\opt$ on every instance, and equal to $\opt$ on $20$ instances. The number of steps $\lb_0$ is on average $0.99$, $0.98$, and $0.97$ times $\opt$, and at least $0.86\,\opt$. For $n\le500$, $\alpha$ is at least $0.987\,\opt$ and on average at least $0.998\,\opt$; CPLEX computes it in at most $88$\,s. Of the gap between $\lb_0$ and $\alpha$, the swaps close $85\%$ and the plateau moves another $14\%$, and $\lb=\alpha$ on $31$ of these $40$ instances. The ratio $\lb/\lb_0$ grows from $1.03$ for $n=1000$ to $1.09$ for $n=10000$.

\begin{figure}[htb]
\centering
\includegraphics[width=\linewidth]{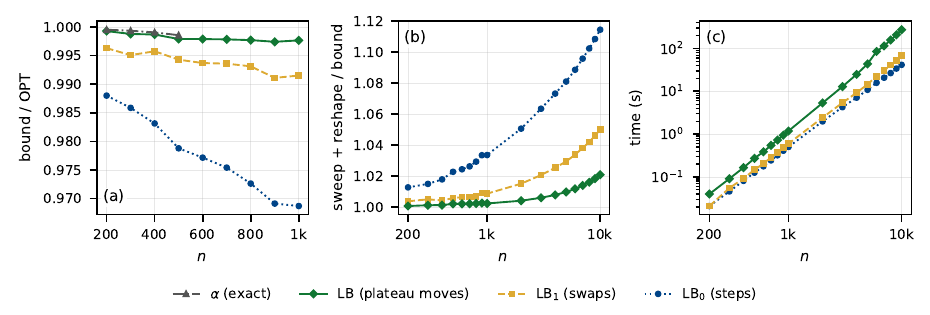}
\caption{The lower bounds in $\R^3$: means over the instances.}
\label{fig:lb3d}
\end{figure}

\subparagraph*{Quality.}
\Cref{fig:summary3d} summarizes the results, and \cref{fig:ratio3d} shows them per input. The sweep is never larger than CK and CK box. It is smaller than both on every instance except the perturbed grids with $n\le1000$, where every dumbbell of every WSPD is a single pair, so that all methods output $\binom n2$ dumbbells. For $n=500$, the sweep is on average $1.4\%$ above $\opt$ (at most $9\%$), and CK box and CK are $15\%$ and $12\%$ above (at most $63\%$ and $36\%$). For $n=1000$, the sweep is at $1.02\,\opt$ and CK at $1.18\,\opt$. For $n=10000$, the sweep is on average at $1.13\,\lb$ (at most $1.23\,\lb$), and CK box and CK at $1.61\,\lb$ and $1.49\,\lb$ (at most $2.03\,\lb$ and $1.68\,\lb$).
\Cref{fig:perpoint3d} gives the number of dumbbells per point of the sweep followed by reshape and the number of pairs per dumbbell of the reference, $\binom n2$ divided by the reference. For $n\le1000$, the uniform, Gaussian, and perturbed grid inputs have at most $1.16$ pairs per dumbbell, so almost all dumbbells are single pairs and every method is close to $\opt$. For $n=10000$, the number of pairs per dumbbell ranges from $1.16$ on the perturbed grid to $244$ on the conical spiral, and it is between $7.6$ and $9.0$ on the scans; the number of dumbbells per point ranges from $21$ on the conical spiral to $4348$ on the perturbed grid.

\begin{figure}[htb]
\centering
\includegraphics[width=\linewidth]{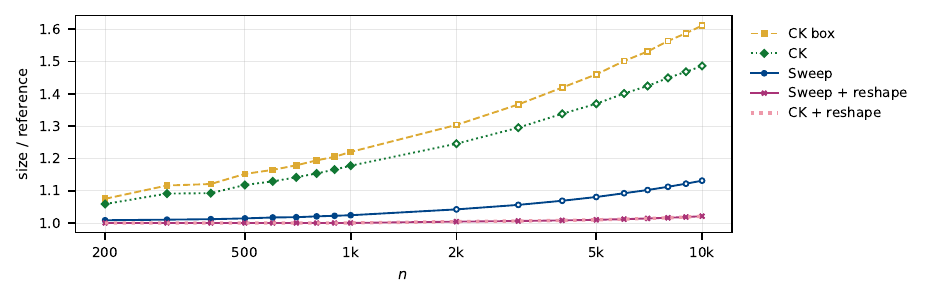}
\caption{Means in $\R^3$. Open markers: the reference is $\lb$ on all instances.}
\label{fig:summary3d}
\end{figure}
\begin{figure}[htb]
\centering
\includegraphics[width=\linewidth]{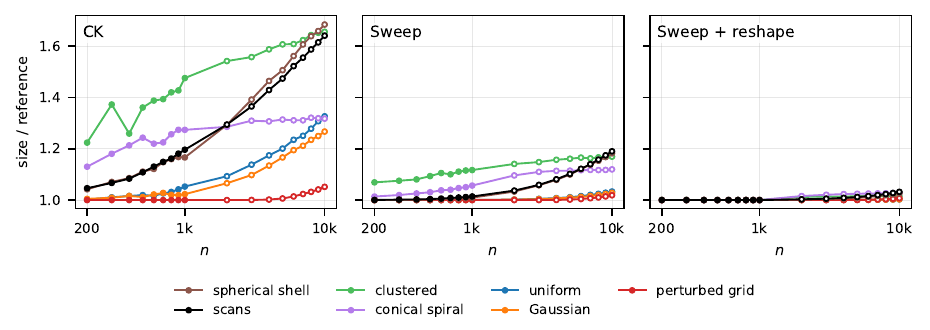}
\caption{Size divided by the reference in $\R^3$, per input. Markers as in \cref{fig:summary3d}.}
\label{fig:ratio3d}
\end{figure}
\begin{figure}[htb]
\centering
\includegraphics[width=\linewidth]{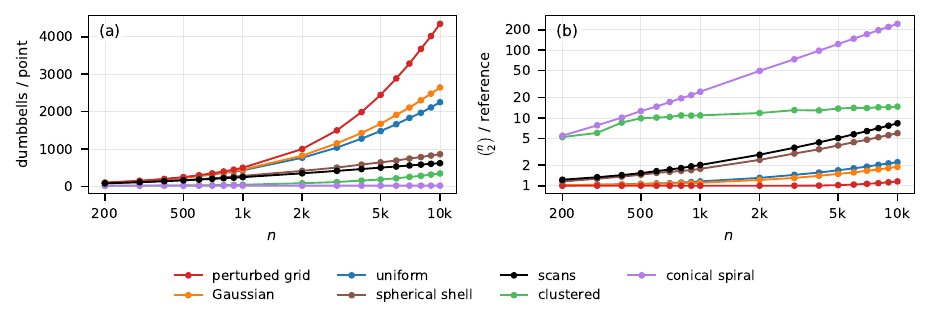}
\caption{Sweep + reshape in $\R^3$: (a)~dumbbells per point, (b)~pairs per dumbbell of the reference, per input.}
\label{fig:perpoint3d}
\end{figure}

\begin{table}[htbp]
\centering\footnotesize\setlength{\tabcolsep}{4pt}
\caption{Size divided by the reference in $\R^3$, per input (one instance each). The column ``Pairs'' gives $\binom n2$ divided by the reference. $\dagger$: the reference is the lower bound. Bold: smallest construction of the row.}\label{tab:exp3d}
\begin{tabular}{@{}lrrccccc@{}}
\toprule
Input & $n$ & Pairs & \shortstack{\strut CK box} & \shortstack{\strut CK} & \shortstack{\strut Sweep} & \shortstack{\strut CK\\\strut + reshape} & \shortstack{\strut Sweep\\\strut + reshape} \\
\midrule
uniform & 1k\phantom{$^\dagger$} & 1.2 & 1.06 & 1.05 & \textbf{1.00} & 1.00 & 1.00 \\
 & 10k$^\dagger$ & 2.2 & 1.37 & 1.33 & \textbf{1.03} & 1.00 & 1.00 \\
\addlinespace[2pt]
Gaussian & 1k\phantom{$^\dagger$} & 1.1 & 1.03 & 1.02 & \textbf{1.00} & 1.00 & 1.00 \\
 & 10k$^\dagger$ & 1.9 & 1.30 & 1.27 & \textbf{1.03} & 1.00 & 1.00 \\
\addlinespace[2pt]
clustered & 1k\phantom{$^\dagger$} & 10.9 & 1.76 & 1.48 & \textbf{1.12} & 1.00 & 1.00 \\
 & 10k$^\dagger$ & 14.7 & 2.03 & 1.66 & \textbf{1.17} & 1.02 & 1.01 \\
\addlinespace[2pt]
perturbed grid & 1k\phantom{$^\dagger$} & 1.0 & \textbf{1.00} & \textbf{1.00} & \textbf{1.00} & 1.00 & 1.00 \\
 & 10k$^\dagger$ & 1.2 & 1.05 & 1.05 & \textbf{1.02} & 1.01 & 1.01 \\
\addlinespace[2pt]
conical spiral & 1k\phantom{$^\dagger$} & 24.4 & 1.28 & 1.27 & \textbf{1.06} & 1.00 & 1.00 \\
 & 10k$^\dagger$ & 243.8 & 1.32 & 1.32 & \textbf{1.12} & 1.03 & 1.03 \\
\addlinespace[2pt]
spherical shell & 1k\phantom{$^\dagger$} & 1.8 & 1.19 & 1.17 & \textbf{1.01} & 1.00 & 1.00 \\
 & 10k$^\dagger$ & 6.0 & 1.87 & 1.68 & \textbf{1.18} & 1.03 & 1.03 \\
\midrule
bunny & 1k\phantom{$^\dagger$} & 1.9 & 1.21 & 1.19 & \textbf{1.01} & 1.00 & 1.00 \\
 & 10k$^\dagger$ & 8.3 & 1.81 & 1.67 & \textbf{1.23} & 1.04 & 1.04 \\
\addlinespace[2pt]
armadillo & 1k\phantom{$^\dagger$} & 2.2 & 1.23 & 1.21 & \textbf{1.02} & 1.00 & 1.00 \\
 & 10k$^\dagger$ & 9.0 & 1.83 & 1.67 & \textbf{1.20} & 1.04 & 1.04 \\
\addlinespace[2pt]
dragon & 1k\phantom{$^\dagger$} & 2.1 & 1.24 & 1.21 & \textbf{1.01} & 1.00 & 1.00 \\
 & 10k$^\dagger$ & 8.5 & 1.78 & 1.62 & \textbf{1.17} & 1.03 & 1.03 \\
\addlinespace[2pt]
Buddha & 1k\phantom{$^\dagger$} & 2.0 & 1.21 & 1.18 & \textbf{1.01} & 1.00 & 1.00 \\
 & 10k$^\dagger$ & 7.6 & 1.76 & 1.61 & \textbf{1.16} & 1.02 & 1.02 \\
\bottomrule
\end{tabular}
\end{table}

\subparagraph*{Running times.}
\Cref{fig:time3d} reports the running times and the memory. On average, the sweep takes $0.5$\,s for $n=1000$ and $45$\,s for $n=10000$, and CK box and CK take $1.4$\,s and $1.5$\,s for $n=10000$. The sweep followed by reshape reaches $1.001\,\opt$ in $1.0$\,s on average for $n=1000$. After the sweep, computing $\lb_1$ takes $26$\,s and computing $\lb$ $3.9$\,min on average for $n=10000$. For $n=10000$, the sweep and reshape use up to $8.4$\,GB and CK up to $5.0$\,GB, on the perturbed grid, whose WSPDs have the most dumbbells.

\begin{figure}[htb]
\centering
\includegraphics[width=\linewidth]{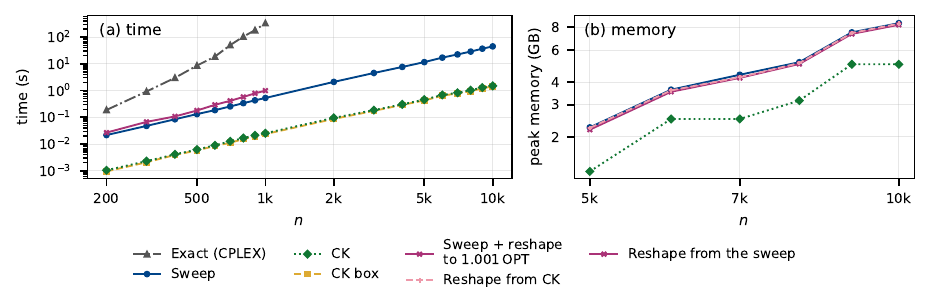}
\caption{Mean running time and largest peak memory in $\R^3$.}
\label{fig:time3d}
\end{figure}

\subparagraph*{Local search.}
Reshape, from either start, is at most $0.03\%$ above $\opt$ on the $90$ instances where $\opt$ is known, and reaches $\opt$ on $72$ of them from CK and on $69$ from the sweep. For $n=10000$, it is on average at $1.021\,\lb$ and at most $1.043\,\lb$, against $1.13\,\lb$ for the sweep alone. The two starts end within $0.03\%$ of each other on every instance. Reshape ends at most as large as LNS from the same start in every run of \cref{fig:ls3d}. For $n=1000$, LNS ends on average $0.36\%$ above the reference from CK and $0.19\%$ from the sweep, against less than $0.01\%$ for reshape, and for $n=10000$, LNS improves CK from $1.55$ to $1.34$ times the reference and the sweep from $1.15$ to $1.10$ (means over the three instances), against $1.02$ for reshape from both starts. \Cref{fig:lstotal3d} shows reshape from both starts against the total time, including the time to build the start.

\begin{figure}[htb]
\centering
\includegraphics[width=\linewidth]{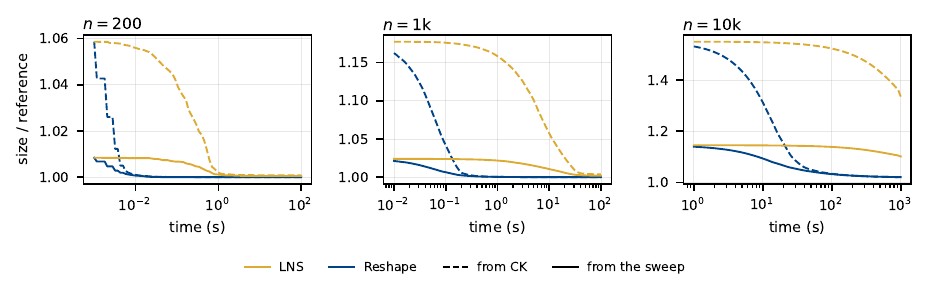}
\caption{Local search in $\R^3$: mean size / reference over time, on all instances for $n\in\{200,1000\}$.}
\label{fig:ls3d}
\end{figure}
\begin{figure}[htb]
\centering
\includegraphics[width=\linewidth]{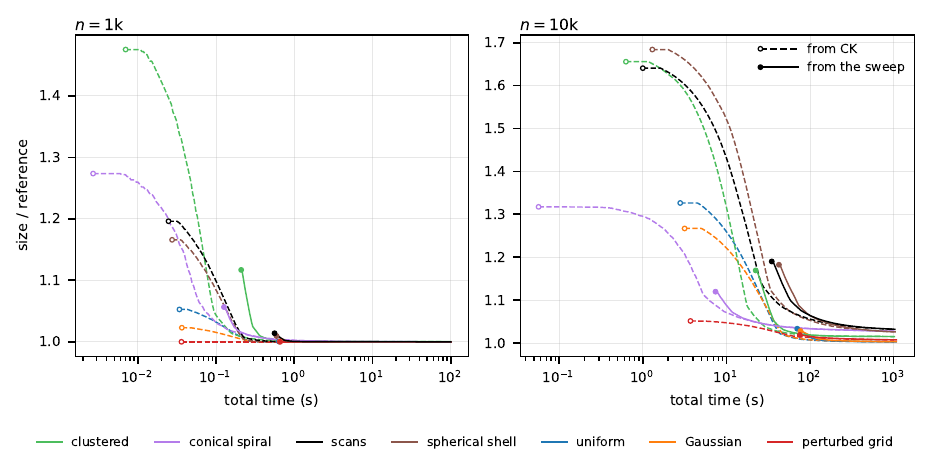}
\caption{Mean size / reference in $\R^3$ per input over the total time: build of the start, then reshape.}
\label{fig:lstotal3d}
\end{figure}

\subparagraph*{Partition WSPDs.}
\Cref{fig:partconv3d} shows the partition WSPDs. For the sweep followed by reshape, the converted partition WSPD is on average $1.005$ and $1.055$ times the reference for $n=1000$ and $10000$, against $1.000$ and $1.021$ for the cover WSPD. CK is at $1.18$ and $1.49$, and at $1.09$ and $1.46$ after the merge pass. It is never smaller than the converted WSPD, also after the merge pass. The octree construction is at $1.51$ and $2.44$, and at $1.14$ and $2.25$ after the merge pass; with the merge pass, it is smaller than CK with the merge pass on $5$ instances, all with $n\le300$.

\begin{figure}[htb]
\centering
\includegraphics[width=\linewidth]{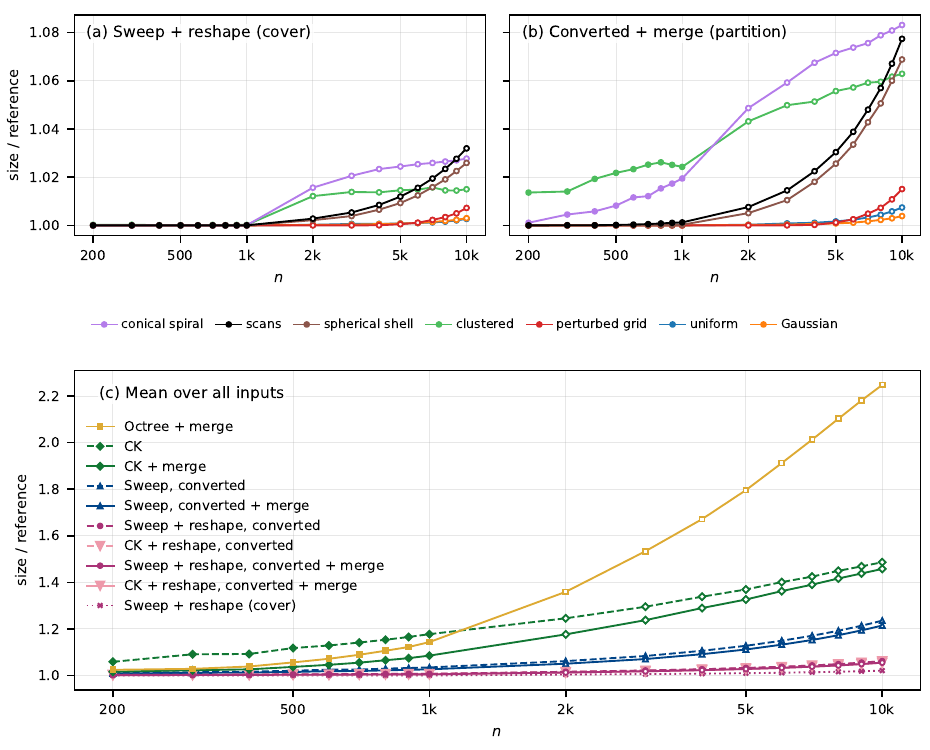}
\caption{Partition WSPDs in $\R^3$: mean size / reference, (a,\,b)~per input. Markers as in \cref{fig:summary3d}.}
\label{fig:partconv3d}
\end{figure}
\clearpage

\section{Larger \texorpdfstring{$\eps$}{epsilon}}
\label{sec:largeeps}

For $\eps>1/16$, the pieces of \cref{sec:sweep} may give dumbbells that are not $\eps$-separated: the proof of \cref{thm:sweep} bounds the distance between two pieces only by $(1-2\eps)D$, which is useless for $\eps\ge1/2$. We show that pieces that depend on $\eps$, together with a separation test, give both sweeps a constant factor for every $\eps>0$. The new ingredient is a lower bound on the distance between two pieces that uses a dumbbell covering the anchor. Throughout, $P\subset\R^2$.

\subparagraph*{Piece families.}
Let $H=\{z\in\ball(0,1): z_1\ge0\}$ be the unit half-disk, where $z_1$ denotes the first coordinate of~$z$. A \emph{piece family} $\mathcal L=(\mathcal L^+,\mathcal L^\circ)$ consists of a partition $\mathcal L^+$ of $H$ and a partition $\mathcal L^\circ$ of $\ball(0,1)$ into finitely many \emph{pieces}. Each piece $\Pi$ comes with numbers $0\le a_\Pi\le b_\Pi\le1$ and $\beta_\Pi>0$ such that $a_\Pi\le\|z\|\le b_\Pi$ for every $z\in\Pi$ and $\diam(\Pi)\le\beta_\Pi$. The \emph{size} of $\mathcal L$ is $|\mathcal L^+|\cdot|\mathcal L^\circ|$. For $0\le a\le b\le1$ and $0\le c\le d\le1$, let
\[
  M_\eps(a,b;c,d)=\min_{r\in[a,b],\ s\in[c,d]}\ \max\Bigl\{1-\eps(r+s),\ r,\ s,\ \frac{1}{1+2\eps}\Bigr\}.
\]
The family is \emph{$\eps$-admissible} if all $\Pi\in\mathcal L^+$ and $\Sigma\in\mathcal L^\circ$ satisfy
\begin{equation}\label{eq:adm}
  \max\{\beta_\Pi,\beta_\Sigma\}\le\max\bigl\{1-\eps(b_\Pi+b_\Sigma),\ M_\eps(a_\Pi,b_\Pi;a_\Sigma,b_\Sigma)-\eps(\beta_\Pi+\beta_\Sigma)\bigr\}.
\end{equation}
Every term on the right-hand side is nonincreasing in $\eps$, so a family that is $\eps_0$-admissible is $\eps$-admissible for every $\eps\le\eps_0$. The first term is the bound used in the proof of \cref{thm:sweep}: the pieces of \cref{sec:sweep}, scaled to radius $1$, with $a_\Pi=0$, $b_\Pi=1$, and the diameters of \cref{lem:pieces}, form a $1/16$-admissible family of size $40$.

\subparagraph*{The collection of a family.}
Fix a piece family $\mathcal L$. For a step with anchor $\{p,q\}$, where $p\prec q$ and $D=\|pq\|$, let $A_\Pi=\{z\in P: z\succeq p,\ z\in p+\eps D\,\Pi\}$ for $\Pi\in\mathcal L^+$, and let $C_\Sigma=P\cap(q+\eps D\,\Sigma)$ for $\Sigma\in\mathcal L^\circ$. The \emph{$\mathcal L$-collection} of the step consists of the dumbbells $\{A_\Pi,C_\Sigma\}$ whose sides are nonempty and disjoint and that are $\eps$-separated. Both conditions are tested on the points, so the $\mathcal L$-collection consists of at most $|\mathcal L^+|\cdot|\mathcal L^\circ|$ dumbbells, all $\eps$-separated. For the pieces of \cref{sec:sweep} and $\eps\le1/16$, every dumbbell $\{A_i,C_k\}$ with nonempty sides passes both tests, so the $\mathcal L$-collection is the fixed collection.

\begin{lemma}[Piece separation]\label{lem:piecesep}
Let $p,q\in\R^2$, $D=\|pq\|$, and let $\{X,Y\}$ be an $\eps$-separated dumbbell with $p,x\in X$ and $q,y\in Y$. Let $A\ni x$ and $C\ni y$ be sets such that $a\eps D\le\|pz\|\le b\eps D$ for every $z\in A$, $c\eps D\le\|qz\|\le d\eps D$ for every $z\in C$, $\diam(A)\le\beta\eps D$, and $\diam(C)\le\beta'\eps D$. Then
\[
  \dist(A,C)\ge D\cdot\max\bigl\{1-\eps(b+d),\ M_\eps(a,b;c,d)-\eps(\beta+\beta')\bigr\}.
\]
\end{lemma}
\begin{proof}
Let $a'\in A$ and $c'\in C$. First, $\|a'c'\|\ge\|pq\|-\|pa'\|-\|qc'\|\ge(1-\eps(b+d))D$. Second, let $\ell=\dist(X,Y)$, $r=\|px\|/(\eps D)$, and $s=\|qy\|/(\eps D)$, so that $r\in[a,b]$ and $s\in[c,d]$. Since $\{X,Y\}$ is $\eps$-separated, $\ell\ge\diam(X)/\eps\ge\|px\|/\eps=rD$, and similarly $\ell\ge sD$. By \cref{obs:local}, $\ell\ge D/(1+2\eps)$. Moreover, $\|xy\|\ge\ell$ and $\|xy\|\ge\|pq\|-\|px\|-\|qy\|=(1-\eps(r+s))D$. Hence $\|xy\|\ge M_\eps(a,b;c,d)\,D$, and $\|a'c'\|\ge\|xy\|-\diam(A)-\diam(C)\ge(M_\eps(a,b;c,d)-\eps(\beta+\beta'))D$.
\end{proof}

\begin{lemma}\label{lem:lcover}
If $\mathcal L$ is $\eps$-admissible, then, in every step, with anchor $\{p,q\}$, the $\mathcal L$-collection covers $R_{pq}$.
\end{lemma}
\begin{proof}
Let $\{x,y\}\in R_{pq}$. By \cref{obs:hered}, we may name $x$ and $y$ so that $\{\{p,x\},\{q,y\}\}$ has disjoint sides and is $\eps$-separated. Its distance is at most $D$, so $\|px\|\le\eps D$ and $\|qy\|\le\eps D$. Since $x\succeq p$ by the definition of $R_{pq}$, the first coordinate of $x$ is at least that of $p$, so $x\in p+\eps D\,H$. Hence $x\in A_\Pi$ for some $\Pi\in\mathcal L^+$, and similarly $y\in C_\Sigma$ for some $\Sigma\in\mathcal L^\circ$. By the definition of a piece family, $A=A_\Pi$ and $C=C_\Sigma$ satisfy the conditions of \cref{lem:piecesep} with $(a,b,\beta)=(a_\Pi,b_\Pi,\beta_\Pi)$ and $(c,d,\beta')=(a_\Sigma,b_\Sigma,\beta_\Sigma)$. By \cref{lem:piecesep} and~\eqref{eq:adm},
\[
  \dist(A_\Pi,C_\Sigma)\ge D\max\{\beta_\Pi,\beta_\Sigma\}\ge\max\{\diam(A_\Pi),\diam(C_\Sigma)\}/\eps .
\]
The left-hand side is positive, so the sides are disjoint, and the dumbbell $\{A_\Pi,C_\Sigma\}$ is $\eps$-separated. Hence it belongs to the $\mathcal L$-collection, and it covers $\{x,y\}$.
\end{proof}

\subparagraph*{The sweeps.}
Let $\mathcal L$ be a piece family of size $c$. The \emph{sweep with $\mathcal L$} outputs the $\mathcal L$-collection in every step. In every step, the \emph{greedy sweep with $\mathcal L$} computes a greedy cover of $R_{pq}$ as in \cref{sec:greedy}, except that, if $|R_{pq}|>16$, it stops after $c$ grown dumbbells. It outputs the greedy cover if the greedy cover covers $R_{pq}$ and has at most as many dumbbells as the $\mathcal L$-collection, and it outputs the $\mathcal L$-collection otherwise. For the pieces of \cref{sec:sweep} and $\eps\le1/16$, these are the fixed sweep of \cref{sec:sweep} and the greedy sweep of \cref{sec:greedy}. As in \cref{sec:greedy}, the cap does not change the output: if $|R_{pq}|>16$, every grown dumbbell belongs to the greedy cover, so a greedy cover with more than $c$ dumbbells is larger than the $\mathcal L$-collection.

\begin{theorem}\label{thm:largeeps}
Let $\eps>0$, let $\mathcal L$ be an $\eps$-admissible piece family of size $c$, and let $P\subset\R^2$. The sweep and the greedy sweep with $\mathcal L$ perform at most $\opt_\eps$ steps and output an $\eps$-WSPD of size at most $c\,\opt_\eps$. For a fixed family and $\eps\ge1/16$, both run in $O(n^3)$ time and $O(n^2)$ space.
\end{theorem}
\begin{proof}
\emph{Size and steps.} Every output dumbbell is $\eps$-separated: the grown dumbbells by construction, and those of the $\mathcal L$-collection by its tests. By \cref{lem:lcover}, every step, with anchor $\{p,q\}$, covers $R_{pq}$, which contains $\{p,q\}$, with at most $c$ dumbbells. As explained before \cref{thm:greedy}, the proofs of validity and of the approximation factor in \cref{thm:sweep} use only this property, with $c$ in place of $40$, and \cref{obs:anchors} bounds the number of steps by $\opt_\eps$.

\emph{Running time.} Let $\eps\ge1/16$. Every $1/16$-separated dumbbell is $\eps$-separated, so $\opt_\eps\le\opt_{1/16}=O(n)$ by the CK construction~\cite{CK95}, and there are $O(n)$ steps. We store the covered pairs in an $n\times n$ table and inspect every pair once, as in the proof of \cref{thm:sweep}. It remains to show that a step takes $O(n^2)$ time. Every point is assigned to its piece in $O(1)$ time, since the family is fixed. The diameters of all sets $A_\Pi$ and $C_\Sigma$ are computed in $O(n^2)$ time. Testing the dumbbell $\{A_\Pi,C_\Sigma\}$ takes $O(|A_\Pi|\cdot|C_\Sigma|)$ time, which sums to $O(n^2)$ over all $\Pi$ and $\Sigma$. Enumerating $R_{pq}$ takes $O(n^2)$ time, since compatibility is tested in $O(1)$ time by \cref{obs:hered}. A dumbbell is grown in $O(n^2)$ time, as in the proof of \cref{thm:greedy}, and at most $\max\{c,16\}$ dumbbells are grown. Marking the pairs covered by each of the at most $c$ output dumbbells takes $O(n^2)$ time. Hence the total time is $O(n^3)$, and the table takes $O(n^2)$ space.
\end{proof}

In particular, the number of steps of either sweep is a lower bound on $\opt_\eps$ for every $\eps>0$.

\subparagraph*{Ring families.}
Radii $0=r_0<r_1<\dots<r_K=1$ and positive integers $n_1,\dots,n_K$ define a partition of $H$: for $j=1,\dots,K$, the $j$th \emph{ring}, the set of points $z\in H$ with $r_{j-1}\le\|z\|\le r_j$, is cut into $n_j$ annular sectors of angle $\pi/n_j$. Similarly, other radii and integers define a partition of $\ball(0,1)$, with sectors of angle $2\pi/n_j$. These pieces form a \emph{ring family}, in which a piece of the $j$th ring has $a=r_{j-1}$ and $b=r_j$, and $\beta$ is its diameter, given by the next lemma. As in \cref{sec:sweep}, points on the boundary between two pieces are assigned to one of them arbitrarily. The pieces of \cref{sec:sweep} form a ring family with $K=2$. \Cref{fig:families} shows the first four families of \cref{tab:families}.

\begin{lemma}\label{lem:sectordiam}
Let $0\le a<b$ and $0<\varphi\le\pi$. The annular sector with radii $a$ and $b$ and angle $\varphi$ (a circular sector if $a=0$) has diameter $\max\{2b\sin(\varphi/2),\sqrt{a^2+b^2-2ab\cos\varphi}\}$.
\end{lemma}
\begin{proof}
Since $\varphi\le\pi$, the convex hull of the sector is bounded by the outer arc, the two radial segments, and the segment between the two inner corners (these corners coincide if $a=0$). Hence its extreme points are the points of the outer arc and the inner corners, and the diameter is attained by two of them. Two points of the outer arc are at distance at most $2b\sin(\varphi/2)$, with equality for the two outer corners. An inner corner and a point of the outer arc at angle $\theta\in[0,\varphi]$ from it are at distance $\sqrt{a^2+b^2-2ab\cos\theta}$, which increases with $\theta$ and is largest for the opposite outer corner. The two inner corners are at distance $2a\sin(\varphi/2)\le2b\sin(\varphi/2)$.
\end{proof}

A central piece of the disk with $n_1=1$ is a disk of diameter $2r_1$. Every other piece of a ring family has angle at most $\pi$.

\begin{proposition}\label{prop:families}
For each $\eps_0$ in \cref{tab:families}, the ring family given there is $\eps_0$-admissible. Hence, for every $\eps\le\eps_0$, the sweep and the greedy sweep with this family output an $\eps$-WSPD of size at most $c(\eps_0)\,\opt_\eps$, where $c(\eps_0)$ is its size.
\end{proposition}
\begin{proof}
In a ring family, $a$, $b$, and $\beta$ depend only on the ring of a piece, so \eqref{eq:adm} has to be checked only for every pair of rings, one of $H$ and one of $\ball(0,1)$. The function minimized in the definition of $M_\eps$ is the maximum of four affine functions of $(r,s)$. Hence its minimum over the rectangle $[a,b]\times[c,d]$ is attained at a corner, at a point of an edge where two of the four functions are equal, or at a point where three of them are equal, and $M_\eps$ is computed exactly by evaluating these finitely many candidates. We checked~\eqref{eq:adm} at $\eps=\eps_0$ by computer, in floating-point arithmetic, for all pairs of rings of the ten families. The smallest slack is $2\cdot10^{-5}$, far above the rounding errors. The second part follows from \cref{thm:largeeps}.
\end{proof}

\begin{table}[htb]
\centering\small
\begin{tabular}{@{}rrlrll@{}}
\toprule
$\eps_0$ & $c(\eps_0)$ & & pieces & outer radii of the rings ($\times1/240$) & pieces per ring\\
\midrule
$1/10$ & $45$ & $H$ & $5$ & $80, 240$ & $1, 4$ \\
 & & $\ball(0,1)$ & $9$ & $80, 240$ & $1, 8$ \\ \addlinespace
$1/8$ & $60$ & $H$ & $6$ & $84, 240$ & $1, 5$ \\
 & & $\ball(0,1)$ & $10$ & $92, 240$ & $1, 9$ \\ \addlinespace
$3/16$ & $128$ & $H$ & $8$ & $118, 240$ & $2, 6$ \\
 & & $\ball(0,1)$ & $16$ & $118, 240$ & $4, 12$ \\ \addlinespace
$1/4$ & $264$ & $H$ & $12$ & $138, 240$ & $3, 9$ \\
 & & $\ball(0,1)$ & $22$ & $32, 158, 240$ & $1, 7, 14$ \\ \addlinespace
$3/8$ & $578$ & $H$ & $17$ & $108, 180, 240$ & $3, 6, 8$ \\
 & & $\ball(0,1)$ & $34$ & $108, 180, 240$ & $6, 12, 16$ \\ \addlinespace
$1/2$ & $1431$ & $H$ & $27$ & $83, 142, 184, 240$ & $3, 6, 9, 9$ \\
 & & $\ball(0,1)$ & $53$ & $47, 119, 175, 240$ & $1, 9, 21, 22$ \\ \addlinespace
$5/8$ & $2850$ & $H$ & $38$ & $54, 100, 146, 190, 240$ & $2, 5, 10, 11, 10$ \\
 & & $\ball(0,1)$ & $75$ & $40, 90, 136, 180, 240$ & $1, 9, 16, 24, 25$ \\ \addlinespace
$3/4$ & $5151$ & $H$ & $51$ & $35, 84, 120, 149, 189, 240$ & $1, 5, 10, 11, 12, 12$ \\
 & & $\ball(0,1)$ & $101$ & $35, 84, 120, 149, 189, 240$ & $1, 10, 20, 22, 24, 24$ \\ \addlinespace
$7/8$ & $9112$ & $H$ & $68$ & $26, 64, 100, 122, 154, 194, 240$ & $1, 4, 10, 12, 14, 14, 13$ \\
 & & $\ball(0,1)$ & $134$ & $26, 64, 96, 122, 154, 190, 240$ & $1, 8, 16, 27, 28, 25, 29$ \\ \addlinespace
$1$ & $14365$ & $H$ & $85$ & $21, 54, 77, 103, 127, 157, 192, 240$ & $1, 4, 7, 13, 15, 15, 14, 16$ \\
 & & $\ball(0,1)$ & $169$ & $11, 48, 75, 102, 126, 156, 192, 240$ & $1, 7, 14, 26, 30, 30, 29, 32$ \\
\bottomrule
\end{tabular}
\caption{Ring families that are $\eps_0$-admissible, and their sizes $c(\eps_0)$. For each $\eps_0$, the two rows give the rings of the half-disk $H$ and of the disk $\ball(0,1)$: the outer radius of each ring, in units of $1/240$, and its number of pieces. For $\eps\le1/16$, the pieces of \cref{sec:sweep} give $c=40$.}
\label{tab:families}
\end{table}

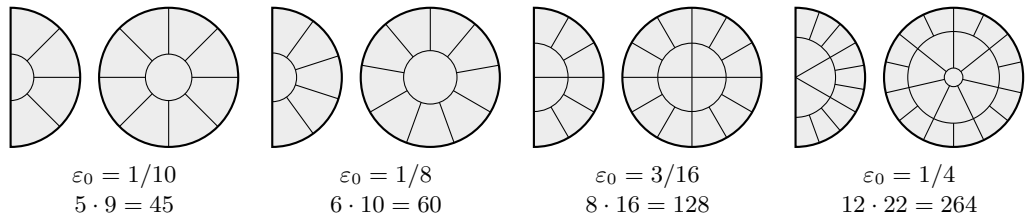
\begin{figure}[htb]
\centering

\begin{tikzpicture}[font=\small]
\fill[black!7] (0,-0.92) arc[start angle=-90, end angle=90, radius=0.92] -- cycle;
\fill[black!7] (2.09,0) circle[radius=0.92];
\draw (0,-0.307) arc[start angle=-90, end angle=90, radius=0.307];
\draw (0.217,-0.217) -- (0.651,-0.651);
\draw (0.307,0) -- (0.92,0);
\draw (0.217,0.217) -- (0.651,0.651);
\draw (2.09,0) circle[radius=0.307];
\draw (2.09,0.307) -- (2.09,0.92);
\draw (1.873,0.217) -- (1.439,0.651);
\draw (1.783,0) -- (1.17,0);
\draw (1.873,-0.217) -- (1.439,-0.651);
\draw (2.09,-0.307) -- (2.09,-0.92);
\draw (2.307,-0.217) -- (2.741,-0.651);
\draw (2.397,-0) -- (3.01,-0);
\draw (2.307,0.217) -- (2.741,0.651);
\draw[thick] (0,-0.92) arc[start angle=-90, end angle=90, radius=0.92] -- cycle;
\draw[thick] (2.09,0) circle[radius=0.92];
\node[anchor=north, align=center] at (1.505,-1.02) {$\eps_0=1/10$\\ $5\cdot9=45$};
\fill[black!7] (3.46,-0.92) arc[start angle=-90, end angle=90, radius=0.92] -- cycle;
\fill[black!7] (5.55,0) circle[radius=0.92];
\draw (3.46,-0.322) arc[start angle=-90, end angle=90, radius=0.322];
\draw (3.649,-0.261) -- (4.001,-0.744);
\draw (3.766,-0.1) -- (4.335,-0.284);
\draw (3.766,0.1) -- (4.335,0.284);
\draw (3.649,0.261) -- (4.001,0.744);
\draw (5.55,0) circle[radius=0.353];
\draw (5.55,0.353) -- (5.55,0.92);
\draw (5.323,0.27) -- (4.959,0.705);
\draw (5.203,0.061) -- (4.644,0.16);
\draw (5.245,-0.176) -- (4.753,-0.46);
\draw (5.429,-0.331) -- (5.235,-0.865);
\draw (5.671,-0.331) -- (5.865,-0.865);
\draw (5.855,-0.176) -- (6.347,-0.46);
\draw (5.897,0.061) -- (6.456,0.16);
\draw (5.777,0.27) -- (6.141,0.705);
\draw[thick] (3.46,-0.92) arc[start angle=-90, end angle=90, radius=0.92] -- cycle;
\draw[thick] (5.55,0) circle[radius=0.92];
\node[anchor=north, align=center] at (4.965,-1.02) {$\eps_0=1/8$\\ $6\cdot10=60$};
\fill[black!7] (6.92,-0.92) arc[start angle=-90, end angle=90, radius=0.92] -- cycle;
\fill[black!7] (9.01,0) circle[radius=0.92];
\draw (6.92,-0.452) arc[start angle=-90, end angle=90, radius=0.452];
\draw (6.92,0) -- (7.372,0);
\draw (7.146,-0.392) -- (7.38,-0.797);
\draw (7.312,-0.226) -- (7.717,-0.46);
\draw (7.372,0) -- (7.84,0);
\draw (7.312,0.226) -- (7.717,0.46);
\draw (7.146,0.392) -- (7.38,0.797);
\draw (9.01,0) circle[radius=0.452];
\draw (9.01,0) -- (9.01,0.452);
\draw (9.01,0) -- (8.558,0);
\draw (9.01,0) -- (9.01,-0.452);
\draw (9.01,0) -- (9.462,-0);
\draw (9.01,0.452) -- (9.01,0.92);
\draw (8.784,0.392) -- (8.55,0.797);
\draw (8.618,0.226) -- (8.213,0.46);
\draw (8.558,0) -- (8.09,0);
\draw (8.618,-0.226) -- (8.213,-0.46);
\draw (8.784,-0.392) -- (8.55,-0.797);
\draw (9.01,-0.452) -- (9.01,-0.92);
\draw (9.236,-0.392) -- (9.47,-0.797);
\draw (9.402,-0.226) -- (9.807,-0.46);
\draw (9.462,-0) -- (9.93,-0);
\draw (9.402,0.226) -- (9.807,0.46);
\draw (9.236,0.392) -- (9.47,0.797);
\draw[thick] (6.92,-0.92) arc[start angle=-90, end angle=90, radius=0.92] -- cycle;
\draw[thick] (9.01,0) circle[radius=0.92];
\node[anchor=north, align=center] at (8.425,-1.02) {$\eps_0=3/16$\\ $8\cdot16=128$};
\fill[black!7] (10.38,-0.92) arc[start angle=-90, end angle=90, radius=0.92] -- cycle;
\fill[black!7] (12.47,0) circle[radius=0.92];
\draw (10.38,-0.529) arc[start angle=-90, end angle=90, radius=0.529];
\draw (10.38,0) -- (10.838,-0.264);
\draw (10.38,0) -- (10.838,0.264);
\draw (10.561,-0.497) -- (10.695,-0.865);
\draw (10.72,-0.405) -- (10.971,-0.705);
\draw (10.838,-0.264) -- (11.177,-0.46);
\draw (10.901,-0.092) -- (11.286,-0.16);
\draw (10.901,0.092) -- (11.286,0.16);
\draw (10.838,0.264) -- (11.177,0.46);
\draw (10.72,0.405) -- (10.971,0.705);
\draw (10.561,0.497) -- (10.695,0.865);
\draw (12.47,0) circle[radius=0.123];
\draw (12.47,0) circle[radius=0.606];
\draw (12.47,0.123) -- (12.47,0.606);
\draw (12.374,0.076) -- (11.996,0.378);
\draw (12.35,-0.027) -- (11.88,-0.135);
\draw (12.417,-0.111) -- (12.207,-0.546);
\draw (12.523,-0.111) -- (12.733,-0.546);
\draw (12.59,-0.027) -- (13.06,-0.135);
\draw (12.566,0.076) -- (12.944,0.378);
\draw (12.47,0.606) -- (12.47,0.92);
\draw (12.207,0.546) -- (12.071,0.829);
\draw (11.996,0.378) -- (11.751,0.574);
\draw (11.88,0.135) -- (11.573,0.205);
\draw (11.88,-0.135) -- (11.573,-0.205);
\draw (11.996,-0.378) -- (11.751,-0.574);
\draw (12.207,-0.546) -- (12.071,-0.829);
\draw (12.47,-0.606) -- (12.47,-0.92);
\draw (12.733,-0.546) -- (12.869,-0.829);
\draw (12.944,-0.378) -- (13.189,-0.574);
\draw (13.06,-0.135) -- (13.367,-0.205);
\draw (13.06,0.135) -- (13.367,0.205);
\draw (12.944,0.378) -- (13.189,0.574);
\draw (12.733,0.546) -- (12.869,0.829);
\draw[thick] (10.38,-0.92) arc[start angle=-90, end angle=90, radius=0.92] -- cycle;
\draw[thick] (12.47,0) circle[radius=0.92];
\node[anchor=north, align=center] at (11.885,-1.02) {$\eps_0=1/4$\\ $12\cdot22=264$};
\end{tikzpicture}
\caption{The ring families of \cref{tab:families} for $\eps_0=1/10$, $1/8$, $3/16$, and $1/4$: the pieces of the half-disk $H$ (left) and of the disk $\ball(0,1)$ (right), and the size $|\mathcal L^+|\cdot|\mathcal L^\circ|$ of the family.}
\label{fig:families}
\end{figure}

\begin{corollary}\label{cor:everyeps}
Let $\eps>0$ and $\beta=\max\{1-2\eps,(1+2\eps)^{-2}\}$. Some $\eps$-admissible piece family has size at most $4\pi^2(1+1/\beta)^4$, which is less than $4\cdot10^5\max\{1,\eps^8\}$. Hence the sweep and the greedy sweep with this family output an $\eps$-WSPD of size at most $4\cdot10^5\max\{1,\eps^8\}\,\opt_\eps$.
\end{corollary}
\begin{proof}
Let the pieces be the nonempty intersections of $H$, and of $\ball(0,1)$, with the half-open squares of a grid of side $\beta/\sqrt2$, with $a_\Pi=0$, $b_\Pi=1$, and $\beta_\Pi=\beta$ for every piece~$\Pi$. If $\beta=1-2\eps$, then the first term of~\eqref{eq:adm} equals $\beta$. Otherwise, $\beta=(1+2\eps)^{-2}$, and since $M_\eps\ge1/(1+2\eps)$, the second term is at least $1/(1+2\eps)-2\eps/(1+2\eps)^2=(1+2\eps)^{-2}=\beta$. Hence the family is $\eps$-admissible. A square that meets $\ball(0,1)$ lies in $\ball(0,1+\beta)$, since its diameter is $\beta$. Comparing areas, at most $\pi(1+\beta)^2/(\beta^2/2)=2\pi(1+1/\beta)^2$ squares meet $\ball(0,1)$, and hence $H$, so the family has size at most $4\pi^2(1+1/\beta)^4$. If $\eps\le1$, then $\beta\ge1/9$ and $4\pi^2(1+1/\beta)^4\le4\pi^2\cdot10^4<4\cdot10^5$. If $\eps\ge1$, then $1+1/\beta=1+(1+2\eps)^2\le10\eps^2$, and the size is less than $4\cdot10^5\eps^8$.
\end{proof}

\section{Quadtree and Fair Split Tree Partition WSPD Constructions}
\label{sec:partition}

We first prove \cref{thm:partition}, on a quadtree construction, and then bound the size of the CK construction on a fair split tree by the same charging (\cref{thm:fairsplit}).

The running times of this section are for dumbbells given as pairs of nodes of a tree, as in~\cite{CK95}. Listing the points of every dumbbell may take $\Omega(n^2)$ time even if $\opt_\eps=n-1$. Consider the points $p_i=(i,(4/\eps)^i)$ of \cref{sec:intro}, and let $p_j$ be the point of largest index in an $\eps$-separated dumbbell. If the side of $p_j$ contained another point $p_l$, its diameter would be at least $(1-\eps/4)(4/\eps)^j$, while its distance to a point $p_i$ of the other side is at most $\|p_ip_l\|\le2(4/\eps)^{j-1}=(\eps/2)(4/\eps)^j$, and $1-\eps/4>\eps^2/2$. Hence this side is $\{p_j\}$, a dumbbell with $m$ points covers $m-1$ pairs, and the dumbbells of every $\eps$-WSPD of these points have at least $\binom n2$ points in total.

The quadtree construction is a variant of the WSPD construction on a compressed quadtree~\cite[Chapter~3]{Har11}. That construction replaces the larger of two cells that are not separated by its children, and tests separation on the cells. Ours replaces both cells, so it outputs pairs of cells of the same level. We analyze it with separation tested on the cells and on the points. For the net-tree construction in doubling metrics, a bound of $O(\opt_\eps)$ is proved in~\cite[Section~4]{BCHKLPRW26} by a similar charging to an optimal WSPD. We make the constant explicit in the plane. The randomized algorithm uses a randomly shifted quadtree~\cite{Aro98}, \cite[Chapter~11]{Har11}, together with a random scale common to all levels, as in the tree embeddings of~\cite{FRT04}.

Since translations preserve separation, we assume that $P\subset[0,L)^2$ for some $L>0$. The construction uses a square $R_0=o+[0,S)^2$ that contains $P$, with $o\in\R^2$ and $S\ge L$. The deterministic algorithm uses $o=(0,0)$ and $S=L$. The randomized algorithm draws $U$ uniformly from $[0,1)$, sets $S=2^{1+U}L$, and draws $o$ uniformly from $(-S/2,0]^2$. Then $R_0$ contains $[0,L)^2$, because $o+S>S/2\ge L$ in each coordinate. For an integer $i\ge0$, let $s_i=2^{-i}S$. The \emph{cells} of level $i$ are the $4^i$ squares $o+\bigl([as_i,(a+1)s_i)\times[bs_i,(b+1)s_i)\bigr)$ with integers $0\le a,b<2^i$. Every point of $R_0$ lies in exactly one cell of each level, and every cell of level $i$ lies in exactly one cell of each level $j\le i$, its \emph{ancestor} of level $j$. For a cell $c$, let $P_c=P\cap c$. For two distinct cells $c$ and $c'$ of the same level such that $P_c$ and $P_{c'}$ are nonempty, the \emph{box test} accepts if the diameter of the cells is at most $\eps$ times the distance between them, and the \emph{exact test} accepts if $\{P_c,P_{c'}\}$ is $\eps$-separated. Fix one of the two tests. A \emph{cell dumbbell} is a set $\{c,c'\}$ of two such cells that the test accepts. Then $\{P_c,P_{c'}\}$ is $\eps$-separated, since $P_c\subseteq c$ and $P_{c'}\subseteq c'$. A cell dumbbell covers the pairs covered by $\{P_c,P_{c'}\}$. A cell dumbbell $\{c,c'\}$ of level $i$ is \emph{maximal} if, for every $j<i$, the ancestors of $c$ and $c'$ of level $j$ do not form a cell dumbbell. Let $\W_\square$ be the collection of the dumbbells $\{P_c,P_{c'}\}$ over all maximal cell dumbbells $\{c,c'\}$. We show that $\W_\square$ satisfies the size bounds of \cref{thm:partition} with both tests, and the running time with the box test.

\begin{figure}[htb]
\centering

\begin{tikzpicture}[font=\small]
\fill[black!13] (0,0) rectangle (0.42,0.42);
\fill[black!13] (0,0.42) rectangle (0.42,0.84);
\fill[black!13] (0,0.84) rectangle (0.42,1.26);
\fill[black!13] (0.42,0) rectangle (0.84,0.42);
\fill[black!13] (0.42,0.42) rectangle (0.84,0.84);
\fill[black!13] (0.42,0.84) rectangle (0.84,1.26);
\fill[black!13] (0.42,1.26) rectangle (0.84,1.68);
\fill[black!13] (0.84,0) rectangle (1.26,0.42);
\fill[wspdblue!45] (0.84,0.84) rectangle (1.26,1.26);
\fill[black!13] (0.84,1.26) rectangle (1.26,1.68);
\fill[black!13] (0.84,1.68) rectangle (1.26,2.1);
\fill[black!13] (1.26,0) rectangle (1.68,0.42);
\fill[black!13] (1.26,0.42) rectangle (1.68,0.84);
\fill[black!13] (1.26,0.84) rectangle (1.68,1.26);
\fill[black!13] (1.26,1.26) rectangle (1.68,1.68);
\fill[black!13] (1.68,0.42) rectangle (2.1,0.84);
\fill[black!13] (1.68,0.84) rectangle (2.1,1.26);
\fill[black!13] (3.78,0.42) rectangle (4.2,0.84);
\fill[black!13] (4.2,0) rectangle (4.62,0.42);
\fill[black!13] (4.2,1.26) rectangle (4.62,1.68);
\fill[black!13] (4.62,0.42) rectangle (5.04,0.84);
\fill[wspdblue!45] (4.62,0.84) rectangle (5.04,1.26);
\draw[black!30, step=0.42] (-0.42,-0.42) grid (2.52,2.52);
\draw[black!30, step=0.42] (3.36,-0.42) grid (5.46,2.1);
\draw[thick] (0,0) rectangle (2.1,2.1);
\draw[wspdorange, thick] (1.764,1.764) -- (2.016,2.016) (1.764,2.016) -- (2.016,1.764);
\draw[wspdorange, densely dotted, thick] (0.42,0.42) -- (1.68,1.68);
\node[wspdorange, anchor=south, inner sep=1pt, font=\scriptsize] at (1.89,2.52) {$>3\sqrt2\,s$ apart};
\draw[wspdblue, thick, shorten <=4pt, shorten >=4pt] (1.05,1.05) -- (4.83,1.05);
\fill[white] (2.82,0.85) rectangle (3.06,1.25);
\draw[wspdblue] (2.78,0.89) -- (2.92,1.25) (2.96,0.85) -- (3.1,1.21);
\fill (0.336,0.315) circle[radius=1.2pt];
\fill (1.722,0.924) circle[radius=1.2pt];
\fill (0.924,1.722) circle[radius=1.2pt];
\fill (0.63,0.672) circle[radius=1.2pt];
\fill (1.092,1.008) circle[radius=1.2pt];
\fill (1.386,1.344) circle[radius=1.2pt];
\fill (0.588,1.302) circle[radius=1.2pt];
\fill (1.344,0.546) circle[radius=1.2pt];
\fill (1.05,0.252) circle[radius=1.2pt];
\fill (0.252,0.966) circle[radius=1.2pt];
\fill (0.168,0.588) circle[radius=1.2pt];
\fill (0.714,0.126) circle[radius=1.2pt];
\fill (1.764,0.63) circle[radius=1.2pt];
\fill (1.47,0.294) circle[radius=1.2pt];
\fill (0.966,1.386) circle[radius=1.2pt];
\fill (0.588,0.924) circle[radius=1.2pt];
\fill (1.512,1.092) circle[radius=1.2pt];
\fill (4.074,0.714) circle[radius=1.2pt];
\fill (4.452,0.378) circle[radius=1.2pt];
\fill (4.746,0.966) circle[radius=1.2pt];
\fill (4.284,1.26) circle[radius=1.2pt];
\fill (4.956,0.588) circle[radius=1.2pt];
\node[anchor=north] at (1.05,-0.441) {$X$: $K_X=17$};
\node[anchor=north] at (4.41,-0.441) {$Y$: $K_Y=5$};
\node[anchor=north] at (2.73,-1.45) {(a)};
\fill[black!18] (7.517,0.372) rectangle (8.137,0.992);
\fill[black!18] (7.517,0.992) rectangle (8.137,1.612);
\fill[black!18] (8.137,-0.248) rectangle (8.757,0.372);
\fill[black!18] (8.137,0.992) rectangle (8.757,1.612);
\fill[black!18] (8.757,0.372) rectangle (9.377,0.992);
\fill[black!18] (8.757,1.612) rectangle (9.377,2.232);
\draw[black!30] (6.277,-0.992) -- (6.277,2.356);
\draw[black!30] (6.153,-0.868) -- (9.501,-0.868);
\draw[black!30] (6.897,-0.992) -- (6.897,2.356);
\draw[black!30] (6.153,-0.248) -- (9.501,-0.248);
\draw[black!30] (7.517,-0.992) -- (7.517,2.356);
\draw[black!30] (6.153,0.372) -- (9.501,0.372);
\draw[black!30] (8.137,-0.992) -- (8.137,2.356);
\draw[black!30] (6.153,0.992) -- (9.501,0.992);
\draw[black!30] (8.757,-0.992) -- (8.757,2.356);
\draw[black!30] (6.153,1.612) -- (9.501,1.612);
\draw[black!30] (9.377,-0.992) -- (9.377,2.356);
\draw[black!30] (6.153,2.232) -- (9.501,2.232);
\draw[wspdblue, thick, fill=wspdblue, fill opacity=0.12] (6.99,-0.062) -- (7.982,-0.372) -- (8.602,-0.372) -- (9.16,0.31) -- (9.16,0.93) -- (8.788,1.798) -- (8.168,1.798) -- (7.238,1.55) -- (6.99,0.558) -- cycle;
\draw[thick, fill=wspdblue, fill opacity=0.3] (7.61,0.558) -- (8.602,0.248) -- (9.16,0.93) -- (8.788,1.798) -- (7.858,1.55) -- cycle;
\fill (7.61,0.558) circle[radius=1.2pt];
\fill (8.602,0.248) circle[radius=1.2pt];
\fill (9.16,0.93) circle[radius=1.2pt];
\fill (8.788,1.798) circle[radius=1.2pt];
\fill (7.858,1.55) circle[radius=1.2pt];
\fill (8.292,1.054) circle[radius=1.2pt];
\fill[wspdorange] (7.517,0.372) circle[radius=1.8pt];
\fill[wspdorange] (7.517,0.992) circle[radius=1.8pt];
\fill[wspdorange] (8.137,-0.248) circle[radius=1.8pt];
\fill[wspdorange] (8.137,0.992) circle[radius=1.8pt];
\fill[wspdorange] (8.757,0.372) circle[radius=1.8pt];
\fill[wspdorange] (8.757,1.612) circle[radius=1.8pt];
\draw[wspdorange, fill=white] (7.517,1.612) circle[radius=1.8pt];
\draw[wspdorange, fill=white] (8.137,0.372) circle[radius=1.8pt];
\draw[wspdorange, fill=white] (8.137,1.612) circle[radius=1.8pt];
\draw[wspdorange, fill=white] (8.757,0.992) circle[radius=1.8pt];
\draw[<->, black!70] (7.61,-1.085) -- node[below, font=\scriptsize] {$w_1$} (9.16,-1.085);
\draw[<->, black!70] (9.532,0.248) -- node[right, font=\scriptsize] {$w_2$} (9.532,1.798);
\draw[black!40, densely dotted] (7.61,-1.085) -- (7.61,0.558);
\draw[black!40, densely dotted] (9.16,-1.085) -- (9.16,0.93);
\draw[black!40, densely dotted] (9.532,0.248) -- (8.602,0.248);
\draw[black!40, densely dotted] (9.532,1.798) -- (8.788,1.798);
\node at (8.385,0.899) {$C$};
\node[wspdblue, font=\scriptsize, anchor=north east] at (7.424,-0.279) {$C+[-s,0]^2$};
\node[anchor=north] at (8.292,-1.45) {(b)};
\end{tikzpicture}
\caption{(a) The cells of level $i$ met by the sides $X$ and $Y$ of an optimal dumbbell are shaded. Every dumbbell of $\W_\square$ charged to $\{X,Y\}$ is determined by a pair of shaded cells, such as the blue pair (the distance is not to scale). Here $\eps\le1/5$ and $m=5$: $X$ meets cells of a $5\times5$ block, but not two opposite corners of it. (b) The lower-left corners of the cells that meet $X$ (filled dots) lie in $C+[-s,0]^2$, where $C$ is the convex hull of $X$. Other points of the grid may lie in it too (hollow dots).}
\label{fig:partition}
\end{figure}
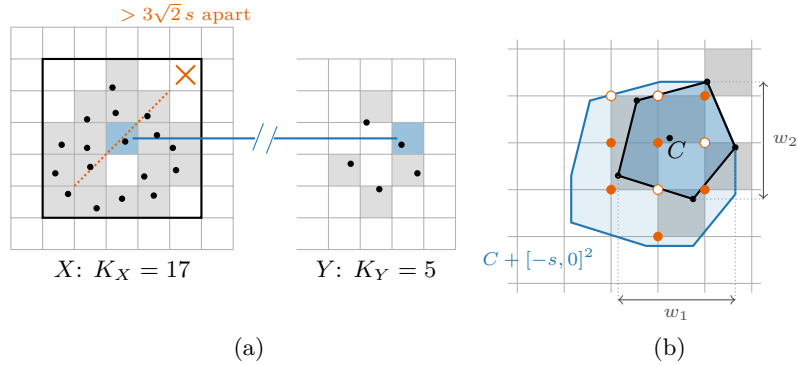

\begin{theorem}\label{thm:partition}
For $P\subset\R^2$ and $\eps\le1/5$, a partition $\eps$-WSPD of size at most $529\,\opt_\eps$ can be computed in $O(n\log n+\opt_\eps)$ time, and a randomized algorithm with the same running time computes one of expected size at most $252\,\opt_\eps$, where a dumbbell is given as a pair of nodes of a compressed quadtree. For $P\subset\R^d$, one of size at most $(\lfloor2\sqrt d(1+2\eps)\rfloor+2)^{2d}\opt_\eps$ can be computed in polynomial time.
\end{theorem}
\begin{proof}
\emph{Partition.} Let $\{x,y\}$ be a pair, and let $c_i$ and $c'_i$ be the cells of level $i$ that contain $x$ and $y$, respectively. A cell dumbbell of level $i$ covers $\{x,y\}$ only if it is $\{c_i,c'_i\}$. For $i$ large enough, $c_i\ne c'_i$ and both tests accept, the exact test because $P_{c_i}=\{x\}$ and $P_{c'_i}=\{y\}$, and the box test because the diameter of the cells tends to $0$ and their distance to $|xy|$. So $\{c_i,c'_i\}$ is a cell dumbbell. Let $i_0$ be the smallest level for which $\{c_{i_0},c'_{i_0}\}$ is a cell dumbbell. It is maximal, and no cell dumbbell $\{c_i,c'_i\}$ with $i>i_0$ is maximal, because its ancestors of level $i_0$ form a cell dumbbell. Hence exactly one maximal cell dumbbell covers $\{x,y\}$. In particular, distinct maximal cell dumbbells cover different pairs and give distinct dumbbells, so $\W_\square$ is a partition $\eps$-WSPD.

\emph{Size.} Fix a minimum $\eps$-WSPD $\W^*$. Charge every dumbbell of $\W_\square$ to a dumbbell of $\W^*$ that covers one of its pairs. Let $Z=\{X,Y\}\in\W^*$, let $\ell=\dist(X,Y)$, and let $\delta=\eps\ell/(1+2\eps)$, so that $\delta=\eps(\ell-2\delta)$. Since $\delta<\ell<\sqrt2L\le\sqrt2S$, there is a level $i\ge1$ whose cells have side $s=s_i$ with $\delta/(2\sqrt2)<s\le\delta/\sqrt2$. Its cells have diameter at most $\delta$. Let $x\in X$ and $y\in Y$, and let $c$ and $c'$ be the cells of level $i$ that contain them. Every point of $c$ is within distance $\delta$ of $x$, and every point of $c'$ is within distance $\delta$ of $y$. Hence $\dist(c,c')\ge\ell-2\delta>0$, so $c\ne c'$, and since $\delta=\eps(\ell-2\delta)$, the box test accepts. Then the exact test accepts as well, so $\{c,c'\}$ is a cell dumbbell. By the first part, the maximal cell dumbbell that covers $\{x,y\}$ is formed by the ancestors of $c$ and $c'$ of some level, so it is determined by $\{c,c'\}$. Therefore, the number of dumbbells charged to $Z$ is at most $K_XK_Y$, where $K_X$ and $K_Y$ are the numbers of cells of level $i$ that meet $X$ and $Y$, respectively (\cref{fig:partition}a).

\emph{Cells met by a small set.} For $D>0$, let $m=\lceil D\rceil+1$, and let $\kappa(D)=m^2-2$ if $D\le(m-2)\sqrt2$ and $\kappa(D)=m^2$ otherwise. We claim that a set $A$ of diameter at most $Ds$ meets at most $\kappa(D)$ cells of level $i$. If two points of $A$ lie in cells whose columns differ by $j\ge1$, their first coordinates differ by more than $(j-1)s$, so $j-1<D$ and $j\le m-1$. Hence $A$ meets at most $m$ columns and, similarly, at most $m$ rows. If it meets fewer than $m$ columns or fewer than $m$ rows, it meets at most $m(m-1)\le m^2-2$ cells. Otherwise, it meets only cells of an $m\times m$ block. Two points in opposite corner cells of this block are at distance more than $(m-2)\sqrt2\,s$, so if $D\le(m-2)\sqrt2$, then $A$ meets at most one cell of each of the two pairs of opposite corners.

\emph{Deterministic bound.} We have $\diam(X)\le\eps\ell=(1+2\eps)\delta<2\sqrt2(1+2\eps)s$, and $2\sqrt2(1+2\eps)<3.96$ for $\eps\le1/5$. For $D=3.96$, we have $m=5$ and $D\le3\sqrt2$, so $K_X\le\kappa(D)=23$. The same holds for $Y$. Hence at most $23^2=529$ dumbbells are charged to $Z$, and $|\W_\square|\le529\,|\W^*|$.

\emph{Randomized bound.} In the randomized algorithm, $s=2^{1+U-i}L$, and the level $i$ depends only on $U$ (for fixed $Z$). Let $V=\log_2\bigl(\delta/(\sqrt2s)\bigr)$. Then $V\in[0,1)$ and $V=\log_2\bigl(\delta/(2\sqrt2L)\bigr)-U+i$, so $V$ is the fractional part of $\log_2\bigl(\delta/(2\sqrt2L)\bigr)-U$, and it is uniformly distributed in $[0,1)$. Let $\gamma=1.4\sqrt2$ and $D_V=\gamma\,2^V$. For $\eps\le1/5$, $\diam(X)\le(1+2\eps)\delta=(1+2\eps)\sqrt2\,2^Vs\le D_Vs$, and the same holds for $Y$, so $K_Y\le\kappa(D_V)$.

Fix $U$. The lower-left corners of the cells of level $i$ lie in the lattice $o+s\mathbb{Z}^2$. Since $i\ge1$, $S/2=2^{i-1}s$ is a multiple of $s$, so each coordinate of $o$ is uniformly distributed modulo $s$, and the lattice is a uniformly random translate of $s\mathbb{Z}^2$. The expected number of its points in a measurable set is the area of the set divided by $s^2$. The cell with lower-left corner $g$ meets $X$ if and only if $g\in X+(-s,0]^2$. This set is contained in $C+[-s,0]^2$, where $C$ is the convex hull of $X$. Let $w_1,w_2\le\diam(X)$ be the widths of $C$ along the two axes. By Cavalieri's principle, adding a horizontal segment of length $s$ to a convex set increases its area by $s$ times its vertical width, and similarly for a vertical segment, so the area of $C+[-s,0]^2$ is $\mathrm{area}(C)+s(w_1+w_2)+s^2$ (\cref{fig:partition}b). By the isodiametric inequality, $\mathrm{area}(C)\le\pi\diam(X)^2/4$. Therefore, $\mathbb{E}[K_X\mid U]\le e(D_V)$, where $e(D)=\pi D^2/4+2D+1$. Since $K_Y\le\kappa(D_V)$, we get $\mathbb{E}[K_XK_Y\mid U]\le\kappa(D_V)\,e(D_V)$. With the substitution $D=\gamma\,2^v$,
\[
  \mathbb{E}[K_XK_Y]\le\int_0^1\kappa(\gamma\,2^v)\,e(\gamma\,2^v)\,dv=\frac1{\ln2}\int_\gamma^{2\gamma}\kappa(D)\Bigl(\frac{\pi D}4+2+\frac1D\Bigr)dD .
\]
The function $\kappa$ equals $9$ on $[\gamma,2]$, $14$ on $(2,2\sqrt2]$, $16$ on $(2\sqrt2,3]$, and $23$ on $(3,2\gamma]$. With the antiderivative $\pi D^2/8+2D+\ln D$, the right side is less than $251.6$. Summing over the dumbbells of $\W^*$ gives $\mathbb{E}|\W_\square|\le252\,|\W^*|$.

\emph{Running time.} We use the compressed quadtree of $P$, which is computed in $O(n\log n)$ time in the real RAM with the floor and logarithm functions~\cite{Har11}. For every node $c$ of the compressed quadtree and every two distinct cells $a,b\subseteq c$ of the next level such that $P_a$ and $P_b$ are nonempty, we call $\textsc{Visit}(a,b)$. Consider a call $\textsc{Visit}(c,c')$. If $P_c$ and $P_{c'}$ are single points, it outputs $\{P_c,P_{c'}\}$. Otherwise, let $j$ be the smallest level greater than that of $c$ at which $P_c$ or $P_{c'}$ meets at least two cells; it is found in the compressed quadtree. From the level of $c$ to level $j-1$, the cells that contain $P_c$ and $P_{c'}$ contain no other point. If the two cells of one of these levels form a cell dumbbell, then so do those of level $j-1$: the exact test gives the same answer at all these levels, and the box test is monotone, because from a level to the next the diameter of the cells halves and their distance does not decrease. Therefore, the call tests the two cells of level $j-1$ and outputs $\{P_c,P_{c'}\}$ if the test accepts. Otherwise, it makes a call $\textsc{Visit}(a,b)$ for every nonempty cell $a\subseteq c$ and every nonempty cell $b\subseteq c'$ of level $j$. Hence the calls output exactly the dumbbells of $\W_\square$, each as a pair of nodes of the compressed quadtree. Every call that outputs nothing makes at least two calls, and there are $O(n)$ initial calls, so there are $O(n+|\W_\square|)$ calls in total. The deterministic bound holds for every square $R_0$ that contains $P$, so $|\W_\square|\le529\,\opt_\eps$ in the randomized algorithm too, and there are $O(n+\opt_\eps)$ calls. With the box test, a call takes $O(1)$ time, since the two cells of level $j-1$ are computed from one point of $P_c$ and one point of $P_{c'}$ with the floor function, and the total time is $O(n\log n+\opt_\eps)$. This is $O(n\log n+n/\eps^2)$, since $\opt_\eps=O(n/\eps^2)$ by the CK construction~\cite{CK95}. With the exact test, a call takes $O(n\log n)$ time, computing the diameters from the convex hulls and the distance from the Delaunay triangulation of $P_c\cup P_{c'}$, which contains a closest pair with one point in each side, and the total time is $O((n+\opt_\eps)\,n\log n)$, which is $O(n^2\min\{n,1/\eps^2\}\log n)$ since also $\opt_\eps\le\binom n2$. Drawing $R_0$ takes $O(1)$ time.

\emph{Higher dimensions.} The deterministic bound extends with cubes of diameter $\sqrt d\,s$ in place of squares: the projection of $X$ on each axis is an interval of length less than $2\sqrt d(1+2\eps)s$, which meets at most $\lfloor2\sqrt d(1+2\eps)\rfloor+2$ of the intervals $[as,(a+1)s)$.
\end{proof}

The exact test never gives more dumbbells than the box test. Every cell dumbbell of the box test is also a cell dumbbell of the exact test, so every maximal cell dumbbell of the box test lies in a unique maximal cell dumbbell of the exact test, formed by ancestors of its cells. Every maximal cell dumbbell of the exact test covers some pair and therefore contains the maximal cell dumbbell of the box test that covers this pair. By Markov's inequality, for $\eps\le1/5$, one run of the randomized algorithm outputs more than $504\,\opt_\eps$ dumbbells with probability at most $1/2$, so the smallest output of $k$ independent runs has size at most $504\,\opt_\eps$ with probability at least $1-2^{-k}$. For $\eps\le1/16$, the randomized bound computed with $\gamma=1.125\sqrt2$ is $141$, with the same pieces of $\kappa$. For every $\eps<1$, the same proofs give the factors $98^2=9604$ (with $D=6\sqrt2$, $m=10$) and $2812$ (with $\gamma=3\sqrt2$). In $\R^d$, the factor is $5^6=15625$ for $d=3$ and $\eps\le1/16$.

\subparagraph*{Fair split trees.}
The charging of \cref{thm:partition} also applies to the CK construction on a fair split tree~\cite{CK95}. We first consider the version used in \cref{sec:experiments}, which has the tree of~\cite{CK95} but differs in the separation test and in the choice of the node that is split, and then the original construction. A \emph{box} is an axis-parallel rectangle. Every node $u$ of the tree $T$ has a point set $P_u\subseteq P$, and $b_u$ denotes the smallest box that contains $P_u$. The root has the point set $P$, and the leaves are the nodes with one point. Every other node $u$ is split by the line through the center of $b_u$ that is perpendicular to a longest side of $b_u$: the point sets of its two children are the points of $P_u$ on each side of the line, where the points on the line go to a fixed side. For two nodes $u$ and $v$ with disjoint point sets, the \emph{box test} accepts if $\max\{\diam(b_u),\diam(b_v)\}\le\eps\dist(b_u,b_v)$, and the \emph{exact test} accepts if the box test accepts or $\{P_u,P_v\}$ is $\eps$-separated. Fix one of the two tests. The procedure $\textsc{Pairs}(u,v)$ outputs $\{P_u,P_v\}$ if the test accepts. Otherwise, it takes the node whose box has the larger diameter, say $u$, breaking ties arbitrarily, and calls $\textsc{Pairs}(u_1,v)$ and $\textsc{Pairs}(u_2,v)$ for the two children $u_1,u_2$ of $u$. This node is not a leaf, since both tests accept two leaves. Let $\W_T$ be the collection of the dumbbells output by the calls $\textsc{Pairs}(u_1,u_2)$ over all nodes of $T$ with children $u_1,u_2$, where a dumbbell is output as a pair of nodes.

\begin{theorem}\label{thm:fairsplit}
For $P\subset\R^2$ and $\eps\le1/5$, the collection $\W_T$ is a partition $\eps$-WSPD of size at most $1024\,\opt_\eps$, both with the box test and with the exact test. With the box test, it is computed in $O(n\log n+\opt_\eps)$ time. For $P\subset\R^d$, its size is at most $\lfloor(2\sqrt d(1+2\eps)+2)^d\rfloor^2\opt_\eps$.
\end{theorem}
\begin{proof}
\emph{Partition.} Every dumbbell of $\W_T$ is $\eps$-separated, because $P_u\subseteq b_u$ for every node $u$. Let $\{x,y\}$ be a pair. The calls $\textsc{Pairs}(u,v)$ with $x\in P_u$ and $y\in P_v$, or vice versa, form a chain: it starts with the call for the two children of the lowest common ancestor of the leaves of $x$ and $y$, every call of the chain that outputs nothing makes exactly one call of the chain, and the chain ends with an output, since two leaves are accepted. Hence exactly one dumbbell of $\W_T$ covers $\{x,y\}$.

\emph{Size.} Fix a minimum $\eps$-WSPD $\W^*$, and charge every dumbbell of $\W_T$ to a dumbbell of $\W^*$ that covers one of its pairs. Let $Z=\{X,Y\}\in\W^*$, and let $\ell=\dist(X,Y)$ and $\delta=\eps\ell/(1+2\eps)$, as in the proof of \cref{thm:partition}. We regard a node as an ancestor of itself. The diameter of the box does not decrease from a node to its parent. For a point $x$, let $u(x)$ be the highest ancestor of the leaf of $x$ whose box has diameter at most $\delta$. The box of the root has diameter at least $\ell>\delta$, so $u(x)$ is not the root, and the box of every proper ancestor of $u(x)$ has diameter greater than $\delta$.

Let $x\in X$ and $y\in Y$. Every point of $b_{u(x)}$ is within distance $\delta$ of $x$, and every point of $b_{u(y)}$ is within distance $\delta$ of $y$. Hence $\dist(b_{u(x)},b_{u(y)})\ge\ell-2\delta>0$, and since $\delta=\eps(\ell-2\delta)$, both tests accept $u(x)$ and $u(y)$. In particular, $P_{u(x)}$ and $P_{u(y)}$ are disjoint, so the first call of the chain of $\{x,y\}$ is on proper ancestors of $u(x)$ and $u(y)$ or on these nodes themselves. We claim that every call $\textsc{Pairs}(u,v)$ of the chain, with $x\in P_u$, is on an ancestor $u$ of $u(x)$ and an ancestor $v$ of $u(y)$. Consider such a call that outputs nothing. Then $u\ne u(x)$ or $v\ne u(y)$. If $u=u(x)$, then $v$ is a proper ancestor of $u(y)$, so $\diam(b_v)>\delta\ge\diam(b_u)$ and the call splits $v$. Similarly, if $v=u(y)$, the call splits $u$. Otherwise, both nodes are proper ancestors. In each case, the next call of the chain is on ancestors of $u(x)$ and $u(y)$. Therefore, the dumbbell of $\W_T$ that covers $\{x,y\}$ is output by a call on ancestors of $u(x)$ and $u(y)$. The calls of the chain depend only on these ancestors, so this dumbbell is determined by $u(x)$ and $u(y)$. The number of dumbbells charged to $Z$ is thus at most $K_XK_Y$, where $K_X$ is the number of distinct nodes $u(x)$ with $x\in X$, and $K_Y$ the number of distinct nodes $u(y)$ with $y\in Y$.

\emph{Packing.} We bound $K_X$ using the outer rectangles of~\cite{CK95}. Let $\lambda_u$ be the length of a longest side of $b_u$. The \emph{outer rectangle} $\hat b_u$ of the root $u$ is a square of side $\lambda_u$ that contains $b_u$. For a child $u$ of a node $v$, $\hat b_u$ is the part of $\hat b_v$ on the side of $P_u$ of the line that splits $v$. Then $b_u\subseteq\hat b_u$, and the outer rectangles of two nodes have disjoint interiors if neither is an ancestor of the other. Moreover, both sides of $\hat b_u$ have length at least $\lambda_v/2$, where $v$ is the parent of $u$. For the sides perpendicular to the line that splits $v$, this holds because $b_v\subseteq\hat b_v$ extends $\lambda_v/2$ on each side of the line. The other sides are as long as those of $\hat b_v$, whose length is $\lambda_v$ if $v$ is the root, and at least $\lambda_{v'}/2\ge\lambda_v/2$ by induction otherwise, where $v'$ is the parent of $v$.

Let $m=\delta/(2\sqrt2)$. The box of the parent $v$ of $u(x)$ has diameter greater than $\delta$, so $\lambda_v>\delta/\sqrt2$ and both sides of $\hat b_{u(x)}$ are longer than $m$. Hence $\hat b_{u(x)}$ contains an axis-parallel square of side $m$ that contains $x$. No node $u(x)$ is a proper ancestor of another one, so these squares, one for each of the $K_X$ nodes, have disjoint interiors. They lie in $C+[-m,m]^2$, where $C$ is the convex hull of $X$. As in the proof of \cref{thm:partition}, the area of this set is $\mathrm{area}(C)+2m(w_1+w_2)+4m^2$, where $w_1,w_2\le\diam(X)$ are the widths of $C$ along the two axes, and $\mathrm{area}(C)\le\pi\diam(X)^2/4$. If $\diam(X)\le Dm$, the area is thus at most $(\pi D^2/4+4D+4)\,m^2$. We have $\diam(X)\le\eps\ell=(1+2\eps)\delta=2\sqrt2(1+2\eps)m$, which is less than $3.96\,m$ for $\eps\le1/5$. Hence $K_X\le\pi D^2/4+4D+4<33$ for $D=3.96$, that is, $K_X\le32$, and the same holds for $K_Y$. At most $32^2=1024$ dumbbells are charged to $Z$, and $|\W_T|\le1024\,|\W^*|$.

\emph{Running time.} The tree is computed in $O(n\log n)$ time~\cite{CK95}. Every call that outputs nothing makes two calls, and there are $n-1$ initial calls, so there are $O(n+|\W_T|)$ calls. With the box test, a call takes $O(1)$ time. Since $|\W_T|\le1024\,\opt_\eps$, the total time is $O(n\log n+\opt_\eps)$, which is $O(n\log n+n/\eps^2)$ since $\opt_\eps=O(n/\eps^2)$~\cite{CK95}.

\emph{Higher dimensions.} In $\R^d$, a node is split by the hyperplane through the center of its box that is perpendicular to a longest side, and the same proof applies with $m=\delta/(2\sqrt d)$: the outer box of $u(x)$ contains a cube of side $m$ that contains $x$, and $X$ lies in a cube of side $\diam(X)\le2\sqrt d(1+2\eps)m$. These cubes lie in a cube of side $\diam(X)+2m$, so $K_X\le(2\sqrt d(1+2\eps)+2)^d$.
\end{proof}

For $\eps\le1/16$, the same proof gives the factor $24^2=576$ (with $D=1.125\cdot2\sqrt2$), and for every $\eps<1$, the factor $94^2=8836$ (with $D=6\sqrt2$). In $\R^3$, the factor is $205^2=42025$ for $\eps\le1/16$. With the exact test, a call takes $O(n\log n)$ time, as in the proof of \cref{thm:partition}.

The proof also applies to the original construction of~\cite{CK95}, with two changes. First, its test, for a separation $s>0$, accepts two nodes $u$ and $v$ if $b_u$ and $b_v$ are contained in two balls of the same radius $r$ that are at distance at least $sr$ from each other. For $s=2/\eps$, this test accepts only $\eps$-separated dumbbells, since $\diam(P_u),\diam(P_v)\le2r$ and $\dist(P_u,P_v)\ge2r/\eps$. It accepts $u(x)$ and $u(y)$: for $r=\max\{\diam(b_{u(x)}),\diam(b_{u(y)})\}/2\le\delta/2$, every point of the balls of radius $r$ centered at the centers of the two boxes is within distance $2r\le\delta$ of $x$ or of $y$, so the distance between the balls is at least $\ell-2\delta=\delta/\eps\ge sr$. Second, $\textsc{Pairs}$ splits the node whose box has the longer side. Then $u(x)$ is defined as the highest ancestor of the leaf of $x$ with $\lambda_{u(x)}\le\delta/\sqrt2$, so that its box has diameter at most $\delta$ and the box of its parent has a side longer than $\delta/\sqrt2$, and the rest of the proof is unchanged. Hence, for $P\subset\R^2$ and $\eps\le1/5$, the construction of~\cite{CK95} with separation $2/\eps$ computes a partition $\eps$-WSPD of size at most $1024\,\opt_\eps$ in $O(n\log n+\opt_\eps)$ time.

\section{Conclusion}
\label{sec:conclusion}

We gave approximation algorithms for \minwspd with explicit factors that keep the separation $\eps$. For cover WSPDs, the sweep is a $40$-approximation in the plane for $\eps\le1/16$, a $2048$-approximation in $\R^3$, and a $2^{O(d)}$-approximation in $\R^d$. With pieces that depend on $\eps$, it keeps a constant factor for every $\eps>0$. For partition WSPDs, CK on a fair split tree and the quadtree construction give factors $1024$ and $529$ in the plane for $\eps\le1/5$, and $252$ in expectation for a randomly shifted and scaled quadtree, in $O(n\log n+\opt)$ time. The sweep also gives a lower bound on $\opt$: its anchors are pairwise incompatible, so its number of steps is at most $\opt$.

On the practical side, the greedy sweep keeps the factor $40$, and reshape shrinks any given WSPD without increasing its size. A local search on the anchors of the sweep enlarges its lower bound to $\lb$, which is at least $0.989\,\opt$ on every planar instance where $\opt$ is known, and certifies the quality of any WSPD on inputs where $\opt$ is out of reach. In the plane, the greedy sweep followed by reshape is within $0.02\%$ of $\opt$ for $n=500$ and on average $1.05\,\lb$ for $n=10000$. Converting its output gives partition WSPDs of on average $1.13\,\lb$, against $1.55\,\lb$ for CK. In $\R^3$, reshape is at most $0.03\%$ above $\opt$ where $\opt$ is known, and on average $1.021\,\lb$ for $n=10000$.

\subparagraph*{Open problem.}
Without the final deletion of redundant dumbbells, which can only decrease the size, the largest ratio we observed between the output of the greedy sweep and the smallest known WSPD is $1.503$, for \texttt{jupiter-0020000} with $\eps=1/20$. Hence its approximation factor for small $\eps$ is between $1.503$ and $40$. What is it?

The largest observed ratio grows with $n$, from $1.41$ for $n=10000$ to $1.50$ for $n=20000$, so the lower end may be far from the truth. The upper end is the product of two bounds: the number of steps is at most $\opt$, and every step outputs at most $40$ dumbbells. An analysis that bounds each step separately cannot improve it much: for $\eps=1/16$, we found point sets in which the set $R_{pq}$ of a single step needs $23$ dumbbells. A smaller factor seems to need an amortized argument, which pays for a step whose $R_{pq}$ needs many dumbbells with optimal dumbbells that cover pairs outside $R_{pq}$.

\bibliography{refs}

\end{document}